\documentclass{llncs}
\usepackage{amsmath}
\usepackage{amsfonts}
\usepackage{amssymb}
\usepackage{mathtools}

\newif\ifeprint
\eprinttrue

\newcommand{\AS}[1]{\textcolor{red}{[{\bf Andr\'{e} S:} #1]}}
\newcommand{\XB}[1]{\textcolor{blue}{[{\bf Xavier B:} #1]}}
\newcommand{\YS}[1]{\textcolor{green}{[{\bf Yixin S:} #1]}}
\newcommand{\AC}[1]{\textcolor{orange}{[{\bf André C:} #1]}}
\renewcommand{\AS}[1]{}
\renewcommand{\XB}[1]{}
\renewcommand{\YS}[1]{}
\renewcommand{\AC}[1]{}

\usepackage[ruled,vlined]{algorithm2e}
\usepackage{algpseudocode}
\usepackage[utf8]{inputenc}
\usepackage{tikz}
\usetikzlibrary{trees}
\usetikzlibrary{patterns}

\usepackage{pgfplots}
\usepackage{bbm}
\usepackage{multirow}
\usepackage{enumerate}
\usepackage{enumitem}
\usepackage{booktabs}
\usepackage{xcolor}
\usepackage{xifthen}

\pgfplotsset{compat=1.10}

\usepgfplotslibrary{fillbetween}
\colorlet{darkgreen}{green!50!black}
\colorlet{darkblue}{blue!50!black}
\colorlet{darkorange}{orange!90!black}
\colorlet{darkred}{red!50!black}

\renewcommand{\vec}[1]{\ensuremath{\mathbf{#1}}}
\newcommand{\bigO}[1]{ \mathcal{O} \left(#1 \right)}
\newcommand{\bigOt}[1]{ \widetilde{\mathcal{O}} \left(#1 \right)}
\newcommand{\bigOg}[1]{\Omega \left(#1 \right)}
\newcommand{\zo}{\{0,1\}}

\newcommand{\bra}[1]{\left\langle #1 \right|}
\newcommand{\ket}[1]{\left| #1 \right\rangle}
\newcommand{\braket}[2]{\left\langle  #1  | #2 \right\rangle}
\newcommand{\ceil}[1]{\left \lceil #1 \right\rceil}
\newcommand{\floor}[1]{\left \lfloor #1 \right\rfloor}

\DeclareMathOperator{\mR}{\mathbb{R}}
\DeclareMathOperator{\mZ}{\mathbb{Z}}
\DeclareMathOperator{\cL}{\mathcal{L}}

\newcommand{\poly}{\operatorname{poly}}

\newcommand{\prob}[2][]{\Pr
\ifthenelse{\isempty{#1}}{}{_{#1}}
\ifthenelse{\isempty{#2}}{}{\left[ #2 \right]}
}

\DeclarePairedDelimiter{\set}{\{}{\}}

\newcommand{\eps}{\varepsilon}
\newcommand{\hh}{\mathcal{H}}
\newcommand{\Unif}{\mathrm{Unif}}
\newcommand{\xh}{\widehat{x}}
\newcommand{\yh}{\widehat{y}}

\newcommand{\pyh}{\widehat{p_y}}
\newcommand{\SWUP}{\textrm{SWUP}}
\newcommand{\Swap}{\textrm{Swap}}
\newcommand{\UP}{\textrm{UP}}
\newcommand{\Enc}{\textrm{Enc}}

\usepackage[colorlinks=true]{hyperref}
\usepackage{cleveref}

\crefname{problem}{problem}{problems}
\Crefname{problem}{Problem}{Problems}
\crefname{proposition}{proposition}{propositions}
\Crefname{proposition}{Proposition}{Propositions}

\renewcommand{\autoref}[1]{\Cref{#1}}

\ifeprint % eprint version

\else % submission version

\fi

\def\epsilon{\varepsilon}

\algblockdefx[Repeat]{Repeat}{EndRepeat}%
[1]{\textbf{Repeat} #1 \textbf{times}}%
{\textbf{EndRepeat}}

\usepackage{xstring}
\makeatletter\def\setkey#1#2#3{\protected@edef#1##1{\protect\IfEq{#2}{##1}{#3}{#1{##1}}}}\makeatother
\def\newdict#1#2{\def#1##1{#2}}

\title{Finding many Collisions via Reusable Quantum Walks}
\subtitle{Application to Lattice Sieving}
\ifeprint
\author{Xavier Bonnetain\inst{1} \and André Chailloux\inst{2} \and André Schrottenloher\inst{3} \and Yixin Shen\inst{4}}
\institute{Université de Lorraine, CNRS, Inria, Nancy, France \and
Inria, Paris, France \and
Cryptology Group, CWI, Amsterdam, The Netherlands \and
Royal Holloway, University of London}
\else
\author{}
\institute{}
\fi
\begin{document}
\maketitle
\renewcommand{\labelitemi}{$\bullet$}

\begin{abstract}
Given a random function $f$ with domain $[2^n]$ and codomain $[2^m]$, with $m \geq n$, a collision of $f$ is a pair of distinct inputs with the same image. Collision finding is an ubiquitous problem in cryptanalysis, and it has been well studied using both classical and quantum algorithms. Indeed, the quantum query complexity of the problem is well known to be $\Theta(2^{m/3})$, and matching algorithms are known for any value of $m$.

The situation becomes different when one is looking for \emph{multiple} collision pairs. Here, for $2^k$ collisions, a query lower bound of $\Theta(2^{(2k+m)/3})$ was shown by Liu and Zhandry (EUROCRYPT~2019). A matching algorithm is known, but only for relatively small values of $m$, when many collisions exist. In this paper, we improve the algorithms for this problem and, in particular, extend the range of admissible parameters where the lower bound is met.

Our new method relies on a \emph{chained quantum walk} algorithm, which might be of independent interest. It allows to extract multiple solutions of an MNRS-style quantum walk, without having to recompute it entirely: after finding and outputting a solution, the current state is reused as the initial state of another walk.

As an application, we improve the quantum sieving algorithms for the shortest vector problem (SVP), with a complexity of $2^{0.2563d + o(d)}$ instead of the previous $2^{0.2570d + o(d)}$.
\end{abstract}

\keywords{Quantum algorithms, quantum walks, collision search, element distinctness, lattice sieving}

% !TeX root = ../collision-walk.tex

\section{Introduction}
\label{sec:intro}

Quantum walks are a powerful algorithmic tool which has been used to provide state-of-the-art algorithms for various important problems in post-quantum cryptography, such as the shortest vector problem (SVP) via lattice sieving~\cite{CL21}, the subset sum problem~\cite{DBLP:conf/asiacrypt/BonnetainBSS20}, information set decoding~\cite{DBLP:conf/pqcrypto/KachigarT17}, etc.

These applications are all established under a particular quantum walk framework called the MNRS framework~\cite{DBLP:journals/siamcomp/MagniezNRS11}, and the quantum walks look for marked nodes in a so-called Johnson graph~\cite{DBLP:conf/pqcrypto/KachigarT17} (or a product of Johnson graphs). When walking on this particular graph, the MNRS framework is somewhat rigid. First, it requires to setup the uniform superposition of all nodes along with their attached data structure, then it applies multiple times reflection operators which move this quantum state close to the uniform superposition of all marked nodes.

%In order to make use of this framework, each node needs to have a history-independent data structure, which should further be efficiently updatable from a node to its neighbors.
%Previously, partial solution was proposed in \cite{DBLP:conf/asiacrypt/BonnetainBSS20}.

Due to this rigidity, previously, the best way to find $k$ different marked nodes was to run the whole quantum walk (including the setup) $k$ times. In~\cite{CL21} the authors noticed that with a better algorithm for finding \emph{multiple} close lattice vectors, instead of one with a single walk,  we would improve the total quantum time complexity of their algorithm for solving the SVP.

A natural observation which guides us throughout this paper is that in certain cases, after obtaining the uniform superposition of all marked nodes via the MNRS quantum walk, it is possible to retrieve part of the solution from the data structure and start another MNRS quantum walk using the remaining part of the quantum state as the new starting state. By doing so, we avoid repeating the setup cost for each new quantum walk, and we now benefit from a trade-off.

In particular, using this observation, we tackle the following problem:

\begin{problem}[Multiple collision search]
Let $f : \zo^n \to \zo^m$, $n \leq m \leq 2n$ be a random function. Let $k \leq 2n-m$. Find $2^k$ \emph {collision pairs}, that is, pairs of distinct $x,y$ such that $f(x) = f(y)$.
\end{problem}

The constraints on the input and output domain are such that a significant ($\Theta\left(2^{2n-m}\right)$) number of collisions pairs exist in the random case. This problem has several applications both in asymmetric and symmetric cryptography. For example, the subproblem in lattice sieving of finding multiple close vectors to a target vector mentioned before can be seen as a special case. The limited-birthday problem, which appears in symmetric cryptanalysis (e.g., impossible differential attacks~\cite{DBLP:conf/asiacrypt/BouraNS14} and rebound distinguishers~\cite{DBLP:journals/tosc/HosoyamadaNS20}) is another example.

\paragraph{Lower Bounds.}
While quantum query lower bounds for the collision problem (with a single solution) had been known for a longer time, Liu and Zhandry proved more recently in~\cite{DBLP:conf/eurocrypt/LiuZ19} a query lower bound in $\bigOg{2^{2k/3 + m/3}}$ to find $2^k$ solutions, which holds for all values of $m \geq n$.

For relatively small values of $k$ and $m$ (actually, $k \leq 3n -2m$, as we explain in~\autoref{subsection:general-case}), the BHT collision search algorithm~\cite{DBLP:conf/latin/BrassardHT98}, allows to reach this bound. Besides this algorithm, Ambainis' algorithm~\cite{DBLP:journals/siamcomp/Ambainis07} uses a quantum walk to find one collision in time $\bigOt{2^{m/3}}$. However, no matching algorithm was known for other values, and in particular, for finding more than 1 collision with $m$ bigger than $1.5n$.

%was given by Liu and Zhandry. However, a matching algorithm was only known for $m = n$ by adapting the BHT algorithm. In this case ($m=n$),   shows a lower bound on quantum time and memory trade-off.

\subsubsection{Contributions.}
Our main contribution in this paper is a \emph{chained quantum walk} algorithm to solve the multiple collision search problem. We formalize the intuitive idea that the output state of a quantum walk can be \emph{reused}, to some extent, as the starting state of another. Our algorithm runs in $\bigOt{2^{\frac23k+\frac{m}3}}$ time and space, in the qRAM model, for any admissible values of $k, n, m$ such that $k \leq \frac{m}{4}$. By combining it with the BHT approach, we can now meet the lower bound over all values of $k, n, m$, except a range of $(k,m)$ contained in $\left[\frac{n}{3}, n\right] \times \left[n , 1.6 n\right]$. Nevertheless, our approach also improves the known complexities in this range.

\begin{quotation}
\textbf{\Cref{thm:main-result}} (\Cref{sec:many_collisions}). Let $f : \zo^n \to \zo^m$, $n \leq m \leq 2n$ be a random function. Let $k \leq \min(2n-m, m/4)$. There exists a quantum algorithm making $\bigO{2^{2k/3 + m/3}}$ quantum queries to $f$ and running in time $\bigOt{ 2^{2k/3 + m/3} }$, that outputs $2^k$ collision pairs of $f$.
\end{quotation}

Using this result, we improve the state-of-the-art time complexity of quantum sieving to solve the SVP in \cite{CL21} from $2^{0.2570d + o(d)}$ to $2^{0.2563d + o(d)}$.
%The previous limitation in our theorem only comes from the large memory requirement. 
We also provide time-memory tradeoffs that are conjectured to be tight~\cite{DBLP:conf/tqc/HamoudiM21}:

\begin{quotation}
\textbf{\Cref{thm:tradeoff}} (\Cref{sec:many_collisions}). Let $f : \zo^n \to \zo^m$, $n \leq m \leq 2n$ be a random function. For all $k \leq \ell \leq \max(2n-m, m/2)$, there exists an algorithm that computes $2^k$ collisions using $\bigOt{2^\ell}$ qubits and $\bigOt{2^{k+m/2-\ell/2}}$ quantum gates and quantum queries to $f$.
\end{quotation}

%Along the way, we give a detailed reenactment of MNRS quantum walks on Johnson graphs, which are commonly used in quantum cryptanalysis, and clarify the operations of the \emph{quantum radix tree} data structure, an unordered set data structure which is very useful for these algorithms. Here we show the main property needed for our chain: after a walk, by extracting a solution from the radix tree and measuring it, we project the tree on an appropriate state for the next walk.

\subsubsection{Organization.}
In \Cref{sec:prelim} we provide several technical preliminaries on quantum algorithms, especially Grover's algorithm. Indeed, an MNRS quantum walk actually emulates a quantum search, and these results are helpful in analyzing the behavior of such a walk. In \Cref{sec:prelim-walk}, we give important details on the MNRS framework, and in particular, the \emph{vertex-coin encoding}, which is a subtlety often omitted from depictions of the framework in the previous literature. In \Cref{sec:many_collisions} we detail our algorithm assuming a suitable quantum data structure is given, and in \Cref{sec:data} we detail the \emph{quantum radix trees}. While there were already proposed in~\cite{DBLP:phd/basesearch/Jeffery14}, we give new (or previously omitted) details relative to the radix tree operations, memory allocation, and how we can efficiently and robustly extract collisions from it. We give a general summary of the multiple collision search problem in \Cref{subsection:general-case} and our applications in \Cref{section:applications}. 

%With all these tools in hand, we detail our new algorithm in \Cref{sec:many_collisions}, and applications in \Cref{section:applications}. 

%In \Cref{sec:data}, we detail a history-independent quantum radix tree data structure. Using this data structure and the idea of reusing the quantum state after measurement as the starting state of the new quantum walk, we provide an algorithm which solves the problem of finding many pairs of collisions in \Cref{sec:many_collisions}. This leads to an improvement in the time complexity of quantum sieving to solve the shortest vector problem in  \Cref{sec:quantum_sieving}. 

% !TeX root = ../collision-walk.tex

\section{Preliminaries}\label{sec:prelim}

In this section, we give some preliminaries on collision search, quantum algorithms and Grover search, which are important for the analysis of quantum walks and their data structures. 
%We note $||$ the concatenation of bit-strings.

\subsection{Collision Search}

In this paper, we study the problem of \emph{collision search} in random functions.

\begin{problem}\label{pb:collision}
Let $f~: \zo^n \to \zo^m$ ($n \leq m$) be a random function. Find a collision of $f$, that is, a pair $(x,y), x \neq y$ such that $f(x) = f(y)$.
\end{problem}

The case $m < n$ can be solved by the same algorithms as the case $m = n$ by reducing $f$ to a subset of its domain. This is why in the following, we focus only on $m \geq n$. The average number of collisions is $\bigO{2^{2n-m}}$. When $m \geq 2n$, we can assume that exactly one collision exists, or none. Distinguishing between these two cases is the problem of \emph{element distinctness}, which is solved by searching for the collision. Regardless of $m$, the collision problem can be solved in:
\begin{itemize}
\item $\Theta(2^{m/2})$ classical time (and queries to $f$). When $m = n$, the problem is the easiest, as it requires only $\bigO{2^{n/2}}$ time and $\mathsf{poly}(n)$ memory using Pollard's rho method. When $m = 2n$, the problem is harder since the best algorithm also uses $\Theta(2^n)$ memory.
\item $\Theta(2^{m/3})$ quantum time (and quantum queries to $f$). A first algorithm was given by Brassard, H{\o}yer and Tapp to reach this for $m = n$~\cite{DBLP:conf/latin/BrassardHT98}, then the lower bound was proven to be $\Omega(2^{m/3})$~\cite{DBLP:journals/jacm/AaronsonS04}, and afterwards, Ambainis solved the \emph{element distinctness} problem (the case $m = 2n$) by a quantum walk algorithm~\cite{DBLP:journals/siamcomp/Ambainis07} which can be adapted for any value of $m$.
% Note that all known quantum algorithms for collisions lie on the time-space trade-off curve $T^2 \times S = 2^m$, and therefore, to reach a time $\bigOt{2^{m/3}}$, they also need $\bigO{2^{m/3}}$ memory.
%When $m = n$, the BHT algorithm \cite{brassard2002quantum} requires $2^{n/3}$ QRACM. When $m = 2n$, Ambainis' algorithm~\cite{DBLP:journals/siamcomp/Ambainis07} for element distinctness requires $2^{2n/3}$ qubits in the QRAQM model.
\end{itemize}

In our case, we want to solve the problem of \emph{multiple collision search}: as there will be expectedly many collisions in the outputs of $f$, we want to find a significant (exponential in $n$) number of them.

\begin{problem}
  Let $f : \zo^n \to \zo^m$, $n \leq m \leq 2n$, $k \leq 2n-m$. Find $2^k$ collisions of $f$.
\end{problem}

Here the state of the art differ classically and quantumly:
\begin{itemize}
\item Classically, it is well known that the problem can be solved for any $m$ and $k$ in $\Theta(2^{(k+m)/2})$ queries (as long as $2^k$ does not exceed the average number of collisions of $f$).
\item Quantumly, Liu and Zhandry~\cite{DBLP:conf/eurocrypt/LiuZ19} gave a query lower bound $\Omega(2^{2k/3 + m/3})$. However, a matching algorithm is only known for small $m$. For example, this lower bound is matched for $m = n$ by adapting the BHT algorithm~\cite{DBLP:conf/eurocrypt/LiuZ19,DBLP:conf/pqcrypto/HosoyamadaSTX19}.
\end{itemize}

\paragraph{On the Memory Complexity.}
For $m = n$, the best known classical algorithm for multiple collision-finding is the parallel collision search (PCS) algorithm by van Oorschot and Wiener~\cite{DBLP:journals/joc/OorschotW99}. It generalizes Pollard's rho method which finds a single collision in $\bigO{2^{n/2}}$ time and $\mathsf{poly}(n)$ memory. Dinur~\cite{DBLP:conf/eurocrypt/Dinur20} showed that in this regime, the time-space trade-off of the PCS algorithm is optimal. Using a restricted model of computation, it can also be shown optimal for larger values of $m$.

Quantumly, a time-space lower bound of $T^3S \geq \bigOg{2^{3k+m}}$ has been shown~\cite{DBLP:conf/tqc/HamoudiM21}. However, the authors conjecture this bound can be improved to $T^2S \geq \bigOg{2^{2k+m}}$. All known quantum algorithms for collisions, including our new algorithms, match this conjectured lower bound.

%\paragraph{Applications.}
%This problem has several applications both in asymmetric and symmetric cryptography, for example the limited-birthday problem that we study in~\autoref{section:applications}.
%\AS{Name several applications?}\AC{There are some in [CNS17]}

\subsection{Quantum Algorithms}
We refer to~\cite{nielsen2002quantum} for an introduction to quantum computation. We write our quantum algorithms in the standard \emph{quantum circuit model}. By default, we use the (universal) Clifford+T gate set, although our complexity analysis remains asymptotic, and we do not detail our algorithms at the gate level.

\paragraph{Memory Models.}
Many memory-intensive quantum algorithms require some kind of \emph{quantum random-access model} (qRAM), which can be stronger than the standard quantum circuit model. One can encounter two types of qRAM:
\begin{itemize}
\item Classical memory with quantum random access (QRACM): a classical memory of size $M$ can be addressed \emph{in quantum superposition} in $\mathsf{polylog}(M)$ operations.
\item Quantum memory with quantum random access (QRAQM): $M$ \emph{qubits} can be addressed \emph{in quantum superposition} in $\mathsf{polylog}(M)$ operations.
%similar to QRACM, but it addresses \emph{qubits}, and not only classical bits.
\end{itemize}

The QRAQM model is required by most quantum walk based algorithms for cryptographic problems, e.g., subset-sum~\cite{DBLP:conf/pqcrypto/BernsteinJLM13,DBLP:conf/asiacrypt/BonnetainBSS20}, information set decoding~\cite{DBLP:conf/pqcrypto/KachigarT17} and the most recent quantum algorithm for lattice sieving~\cite{CL21}. Previous papers considered that it required only to augment the set of gates available in a quantum circuit with the following ``qRAM read'' gate, which accesses in superposition an array of $M$ memory cells (e.g., individual bits):
\begin{equation}
\ket{y_1, \ldots, y_M} \ket{x} \ket{i} \xmapsto{\mathsf{qRAMR}} \ket{y_1, \ldots, y_M} \ket{x \oplus y_i} \ket{i} \enspace.
\end{equation}

Thus, the term ``QRAQM model'' was deemed equivalent to  ``counting the complexity in Clifford+T+\textsf{qRAM} gates'' instead of just Clifford+T gates. Note that the \textsf{qRAMR} gate is indeed a unitary, and it can be simulated with Clifford+T gates, at the expense of a complexity linear in the number of addressed qubits.

\paragraph{qRAM write.}
The \textsf{qRAMR} gate, as its name indicates, allows only to \emph{read} in superposition. But in most previous algorithms that required the QRAQM model, including the quantum walk algorithms that we are interested in, and those that we introduce in this paper, we actually need a second qRAM gate that we name the ``qRAM write''\footnote{The \textsf{qRAMW} is actually the \textsf{qRAMR} gate composed with a Hadamard transform.}:
\begin{equation}
\ket{y_1, \ldots, y_M} \ket{x} \ket{i} \xmapsto{\mathsf{qRAMW}} \ket{y_1, \ldots, y_i \oplus x, \ldots y_M} \ket{x} \ket{i} \enspace .
\end{equation}
This operation is required to efficiently maintain quantum data structures such as the \emph{quantum radix trees} from the literature (see~\autoref{sec:data}). Indeed, when working on these data structures, it is required to update in polynomial time the data at a position which can be in superposition (e.g., adding a new node to the radix tree). 
%Although the \textsf{qRAMW} gate can appear quite powerful at first sight, it is actually equivalent to the \textsf{qRAMR} gate, up to a Hadamard transform:
%\begin{equation}
%\mathsf{qRAMW} = \left( H^{\otimes M} \otimes H \otimes I \right) \mathsf{qRAMR} \left( H^{\otimes M} \otimes H \otimes I \right)  \enspace,
%\end{equation}
%and it can also be simulated with about the same cost as the \textsf{qRAMR} gate.
In the following, we count the complexity of our algorithms asymptotically on the ``Clifford + T + \textsf{qRAMR} + \textsf{qRAMW}'' gate set, so we assume that both the \textsf{qRAMR} and \textsf{qRAMW} have unit cost, as would be required by previous works.

\paragraph{Collision Finding without qRAM.}
To date, the best quantum algorithms for collision finding, and the ones that reach the query lower bound, require the qRAM model: the BHT algorithm~\cite{DBLP:conf/latin/BrassardHT98} uses QRACM and Ambainis' quantum walk uses QRAQM~\cite{DBLP:journals/siamcomp/Ambainis07} to define time-efficient quantum data structures. Initially Ambainis used a \emph{skip list}. We will focus on the more recent \emph{quantum radix tree}, but the QRAQM requirement remains the same.

To some extent, it is possible to get rid of qRAM. For $m = n$, the complexity rises from $\bigO{2^{m/3}}$ to $\bigO{2^{2m/5}}$~\cite{DBLP:conf/asiacrypt/ChaillouxNS17}. For $m = 2n$, the complexity rises to $\bigO{2^{3m/7}}$~\cite{DBLP:conf/sacrypt/JaquesS20}. These algorithms can also be adapted for multiple collision finding, where they will outperform the classical ones for some parameter ranges (but not all).

\subsection{Grover's Algorithm}\label{subsection:grover}

In this section, we recall Grover's quantum search algorithm~\cite{DBLP:conf/stoc/Grover96} and give a few necessary results for the rest of our analysis. Indeed, as shown in~\cite{DBLP:journals/siamcomp/MagniezNRS11}, an MNRS quantum walk actually emulates a quantum search, up to some error. If we manage to put this error aside, the analysis of the walk follows from the following lemmas.

\paragraph{Original Quantum Search.}
In the original setting of Grover's search, we have a function $g : \zo^{n} \rightarrow \zo$ and the goal is to find $x$ st. $g(x) = 1$ using queries to $g$. In the quantum setting, we have access to the unitary $O_g : \ket{x}\ket{b} \rightarrow \ket{x}\ket{b \oplus g(x)}$, which is an efficient quantum unitary if $g$ is efficiently computable. In particular we can compute 
$\ket{\psi_U} = \frac{1}{\sqrt{2^n}}\sum_{x \in \zo^n} \ket{x}\ket{g(x)}$ with a single call to $O_g$. Let $\eps = \frac{|\{x : g(x) = 1\}|}{2^n}$. We also define the normalized states
\[
\ket{\psi_{B}}  = \frac{1}{\sqrt{(1-\eps)2^n}}\sum_{x: g(x) = 0 } \ket{x}\ket{g(x)}, \quad
\ket{\psi_{G}}  = \frac{1}{\sqrt{\eps 2^n}}\sum_{x: g(x) = 1 } \ket{x}\ket{g(x)}
\]
and $\ket{\psi_U} = \sqrt{1-\eps}\ket{\psi_{B}} + \sqrt{\eps}\ket{\psi_{G}}$. Let $\hh = \mathrm{span}(\{\ket{\psi_B},\ket{\psi_G}\})$. Let $\mathrm{Rot}_{\theta}$ be the $\theta$-rotation unitary in $\hh$:
\begin{equation*}
	\mathrm{Rot}_\theta (\cos(\alpha)\ket{\phi_B} + \sin(\alpha)\ket{\psi_G}) = \cos(\alpha + \theta)\ket{\psi_B} + \sin(\alpha + \theta)\ket{\psi_G} \enspace.
\end{equation*}

For a fixed $\eps$, let $\alpha = \arcsin(\sqrt{\eps})$ so that 
$$\ket{\phi_U} = \sqrt{1-\eps}\ket{\psi_B} + \sqrt{\eps}\ket{\psi_G} = \cos(\alpha)\ket{\psi_B} + \sin(\alpha)\ket{\psi_G},$$
For a state $\ket{\psi} \in \hh$, let $\mathrm{Ref}_{\ket{\psi}}$ be the reflection over $\ket{\psi}$ in $\hh$: 
\begin{align*}
	Ref_{\ket{\psi}} \ket{\psi}  = \ket{\psi} \quad \textrm{and} \quad
	Ref_{\ket{\psi}} \ket{\psi^\bot}  = - \ket{\psi^\bot} 
\end{align*}
where $\ket{\psi^\bot}$ is any state in $\hh$ orthogonal to $\ket{\psi}$\footnote{For a fixed $\ket{\psi}$, $\ket{\psi^\bot}$ is actually unique up to a global phase.}  We have
\begin{align*}\label{Eq:RefToRot}
	\mathrm{Ref}_{\ket{\psi_U}}\mathrm{Ref}_{\ket{\psi_{B}}} = \mathrm{Rot}_{2\alpha} \enspace.
\end{align*}
Assume that we have access to a \emph{checking oracle} $O_{\mathrm{check}}$ which performs:
\begin{align*}
\begin{cases}
	O_{\mathrm{check}} \ket{\psi_{B}}\ket{0} & = \ket{\psi_{B}}\ket{0} \\
	O_{\mathrm{check}} \ket{\psi_{G}}\ket{0} & = \ket{\psi_{G}}\ket{1}
\end{cases}
\end{align*}

In the standard setting described above, this is just copying the last register. Starting from an ``initial state'' $\ket{\psi_U}$, we apply repeatedly an iterate consisting of a reflection over $\ket{\psi_U}$, and a reflection over $\ket{\psi_B}$. This progressively transforms the current state into the ``good state'' $\ket{\psi_G}$. Typically $\mathrm{Ref}_{\ket{\psi_U}}$ is constructed from a circuit that computes $\ket{\psi_U}$ and  $\mathrm{Ref}_{\ket{\psi_{B}}}$ is implemented using the checking oracle above: in that case, we are actually performing an \emph{amplitude amplification}~\cite{brassard2002quantum}.

\begin{proposition}[Grover's algorithm, known $\alpha$]\label{Proposition:GroverOriginal}
	Consider the following algorithm, with $\alpha \le \pi/4$:
	\begin{enumerate}
		\item Start from $\ket{\psi_U}$. 
		\item Apply $\mathrm{Rot}_{2\alpha} = \mathrm{Ref}_{\ket{\psi_U}}\mathrm{Ref}_{\ket{\psi_{B}}}$ $N$ times on $\ket{\psi_U}$ with $N = \lfloor\frac{\pi/2 - \alpha}{2\alpha}\rfloor$.
		\item Apply $O_{\mathrm{check}}$ and measure the last qubit.
	\end{enumerate}
	This procedure measures $1$ wp. at least $1 - 4\alpha^2$ and the resulting state is $\ket{\psi_{G}}$.
\end{proposition}
\begin{proof}
	Let us define $\gamma = \alpha + 2N\alpha$. We have 
	$$(Rot_{2\alpha})^n \ket{\psi_U} = \cos(\alpha + 2N\alpha)\ket{\psi_B} + \sin(\alpha + 2N\alpha)\ket{\psi_G} = \cos(\gamma)\ket{\psi_B} + \sin(\gamma)\ket{\psi_G}.$$
	Notice that we chose $N$ st. $\gamma \le \frac{\pi}{2} < \gamma + 2\alpha$ so $\frac{\pi}{2} - \gamma \in [0,2\alpha)$. After applying the checking oracle, we obtain the state 
	$$ \cos(\gamma)\ket{\psi_B}\ket{0} + \sin(\gamma)\ket{\psi_G}\ket{1}.$$
	Measuring the last qubit gives outcome $1$ with probability $\sin^2(\gamma)$ and the resulting state in the first register is $\ket{\psi_G}$. In order to conclude, we compute
	$$ \sin^2(\gamma) = \cos^2(\pi/2 - \gamma) \ge \cos^2(2\alpha) \ge 1 - 4\alpha^2. \qquad \ensuremath{\Box} $$
\end{proof}		

%In Grover's algorithm, the operation $\mathrm{Ref}_{\ket{\psi_{B}}}$ is implemented by a call to the checking oracle, and the operation $\mathrm{Ref}_{\ket{\psi_U}}$ by an operation as costly as building the initial state $\ket{\psi_U}$. 
%In an MNRS quantum walk, the  operation $\mathrm{Ref}_{\ket{\psi_U}}$ is simulated without having to redo the whole \emph{setup} of the walk (the operation that builds $\ket{\psi_U}$). This is where the framework differs from a simple quantum search.

In our algorithms, we will start not from the uniform superposition $\ket{\psi_U}$, but from the \emph{bad subspace} $\ket{\psi_B}$. We show that this makes little difference.

\begin{proposition}[Starting from $\ket{\psi_B}$, known $\alpha$]\label{Proposition:GroverFromBad}
	Consider the following algorithm, with $\alpha \le \pi/4$:
	\begin{enumerate}
		\item Start from $\ket{\psi_B}$. 
		\item Apply $\mathrm{Rot}_{2\alpha} = \mathrm{Ref}_{\ket{\psi_U}}\mathrm{Ref}_{\ket{\psi_{B}}}$ $N'$ times on $\ket{\psi_B}$ with $N' = \lfloor\frac{\pi/2}{2\alpha}\rfloor$.
		\item Apply the checking oracle and measure the last qubit.
	\end{enumerate}
	This procedure measures $1$ wp. at least $1 - 4\alpha^2$ and the resulting state is $\ket{\psi_{G}}$.
\end{proposition}
\begin{proof}
	The proof is essentially the same as the previous one. 
	Let $\gamma' =  2N'\alpha$. We have 
	$$(Rot_{2\alpha})^{N'} \ket{\psi_B} = \cos(2N'\alpha)\ket{\psi_B} + \sin(2N'\alpha)\ket{\psi_G} = \cos(\gamma')\ket{\psi_B} + \sin(\gamma')\ket{\psi_G}.$$
	Notice that we chose $N'$ st. $\gamma' \le \frac{\pi}{2} < \gamma' + 2\alpha$ so $\frac{\pi}{2} - \gamma' \in [0,2\alpha)$. After applying the checking oracle, we obtain the state 
	$$ \cos(\gamma')\ket{\psi_B}\ket{0} + \sin(\gamma')\ket{\psi_G}\ket{1}.$$
	Measuring the last qubit gives $1$ wp. $\sin^2(\gamma')$ and the resulting state in the first register is $\ket{\phi_G}$. In order to conclude, we compute 
	$$ \sin^2(\gamma') = \cos^2(\pi/2 - \gamma') \ge \cos^2(2\alpha) \ge 1 - 4\alpha^2. \qquad \ensuremath{\Box}  $$
\end{proof}

After applying the check and measuring, if we don't succeed, we obtain the state $\ket{\psi_B}$ again. So we can run the quantum search again.

In Grover's algorithm, we have a procedure to construct $\ket{\psi_U}$ and we use this procedure to initialize the algorithm and to perform the operation $\mathrm{Ref}_{\ket{\psi_U}}$. A quantum walk will have the same general structure as Grover's algorithm, but we will manipulate very large states $\ket{\psi_U}$. Though $\ket{\psi_U}$ is long to construct (the \emph{setup} operation), performing $\mathrm{Ref}_{\ket{\psi_U}}$ will be less costly. 

In the MNRS framework, $\ket{\psi_U}$ is chosen as the unique eigenvector of eigenvalue $1$ of an operator related to a random walk in a graph. To perform $\mathrm{Ref}_{\ket{\psi_U}}$ efficiently, we use phase estimation on this operator.

%To do this, we will choose $\ket{\psi_U}$ as the unique eigenvector of eigenvalue $1$ of an operator related to a random walk in a graph and so we will be able to perform phase estimation to perform $Ref_{\ket{\psi_U}}$ more efficiently. This is the core idea of the MNRS framework of quantum walks that we will now present.

% normally everything that we need about quantum walks is now in this file
% and the rest is legacy
% !TeX root = ../collision-walk.tex

\section{Quantum Walks for Collision Finding}\label{sec:prelim-walk}

In this section, we present MNRS quantum walks, which underlie most cryptographic applications of quantum walks to date, and give important details on their actual implementation using a \emph{vertex-coin encoding}.

\subsection{Definition and Example}

We consider a regular, undirected graph $G = (V, E)$, which in cryptographic applications (e.g., collision search), is usually a Johnson graph (as in this paper) or a product of Johnson graphs (a case detailed e.g. in~\cite{DBLP:conf/pqcrypto/KachigarT17}). 

\begin{definition}[Johnson graph]
The Johnson graph $J(N, R)$ is a regular, undirected graph whose vertices are the subsets of $[N]$ containing $R$ elements, with an edge between two vertices $v$ and $v'$ iff $|v\cap v'|=R-1$. In other words, $v$ is adjacent to $v'$ if $v'$ can be obtained from $v$ by removing an element and adding an element from $[N] \backslash v$ in its place.
\end{definition}

%\begin{definition}[Product of graphs]
%	Let $G_1=(\mathcal{V}_1,\mathcal{E}_1)$ and $G_2=(\mathcal{V}_2,\mathcal{E}_2)$ be two graphs, their product $G_1\times G_2$ is the graph $G=(\mathcal{V},\mathcal{E})$ where: 
%	\begin{enumerate}
%		\item $\mathcal{V}=\mathcal{V}_1\times \mathcal{V}_2$, i.e. $\mathcal{V}=\{v_1v_2|v_1\in\mathcal{V}_1,v_2\in\mathcal{V}_2\} $
%		\item $\mathcal{E}=\{(u_1u_2,v_1v_2)|(u_1=v_1 \wedge (u_2,v_2)\in \mathcal{E}_2) \vee ((u_1,v_1)\in \mathcal{E}_1 \wedge u_2=v_2 )\}$
%	\end{enumerate}
%	
%\end{definition}
%
%\begin{definition}[Product Johnson graph]
%Let $X_1, \ldots, X_t$ be sets of size $N_1, \ldots, N_t$. The product Johnson graph $G = J(X_1, \ldots, X_t ; R_1, \ldots, R_t)$ is the product graph $J(X_1,R_1)\times \ldots \times J(X_t,R_t)$.
%
%We identify a vertex $u = (u_1, \ldots, u_t)$ with a tuple of subsets of $X_1 \times \ldots \times X_t$.
%\end{definition}

In collision search, a vertex in the graph specifies a set of $R$ inputs to the function $f$ under study, where its domain $\zo^n$ is identified with $[2^n]$. Let $M \subseteq V$ be a set of \emph{marked} vertices, e.g., all the subsets $R \subseteq \zo^n$ which contain a collision of $f$: $\exists x,y \in R, x \neq y, f(x) = f(y)$. A classical \emph{random walk} on $G$ finds a marked vertex using~\Cref{alg:random-walk}.

\begin{algorithm}[htbp]
\DontPrintSemicolon
\textbf{Setup} an arbitrary vertex $x \in V$\;
\Repeat{the current vertex is marked}{
\Repeat{the current vertex is uniformly random}{
\textbf{Update}: move to a random adjacent vertex\;
}
\textbf{Check} if the current vertex is marked\;
}
\caption{Classical random walk on $G$}\label{alg:random-walk}
\end{algorithm}

The quantum walk is analogous to this process. Let $\epsilon = \frac{|M|}{|V|}$ be the proportion of marked vertices and $\delta$ be the spectral gap of $G$. Starting from any vertex, after $\bigO{\frac{1}{\delta}}$ updates, we sample a vertex of the graph uniformly at random. For a Johnson graph $J(N,R)$, $\delta = \frac{N}{R(N-R)} \simeq \frac{1}{R}$. Let $\mathsf{S}$ be the time to \textbf{Setup}, $\mathsf{U}$ the time to \textbf{Update}, $\mathsf{C}$ the time to \textbf{Check} a given vertex. Then~\Cref{alg:random-walk} finds a marked vertex in time: $\bigO{ \mathsf{S} + \frac{1}{\epsilon} \left( \frac{1}{\delta} \mathsf{U} + \mathsf{C} \right) }$. Magniez \emph{et al.}~\cite{DBLP:journals/siamcomp/MagniezNRS11} show how to translate this generically in the quantum setting, provided that quantum analogs of the \textbf{Setup}, \textbf{Update} and \textbf{Check} can be implemented.

\begin{theorem}[From~\cite{DBLP:journals/siamcomp/MagniezNRS11}]\label{thm:qwalk}
Assume that quantum algorithms for \textbf{Setup}, \textbf{Update} and \textbf{Check} are given. Then there exists a quantum algorithm that finds a marked vertex in time: $ \bigOt{ \mathsf{S} + \frac{1}{\sqrt{\epsilon}} \left( \frac{1}{\sqrt{\delta}} \mathsf{U} + \mathsf{C} \right) } $
instead of $\bigO{ \frac{1}{\sqrt{\epsilon}} \left(\mathsf{S}+\mathsf{C}\right)}$ with a naive search.
\end{theorem}

%The MNRS framework can be seen as a generalization of quantum search. In fact, it relies on the definition of a \emph{walk step} operator that, within a time complexity $\bigOt{\frac{1}{\sqrt{\delta}} \mathsf{U} }$, emulates a reflection through the setup state (that is, half of a quantum search iterate). Magniez \emph{et al.} used a \emph{phase estimation} algorithm to run this emulation in the most generic way possible, and show how to perform it with an arbitrary precision, up to a polynomial increase in the complexity. Before that, Ambainis~\cite{DBLP:journals/siamcomp/Ambainis07} designed a quantum walk algorithm for the \emph{element distinctness problem} that can be seen as a simple example of an MNRS quantum walk.

%\begin{theorem}[\cite{DBLP:journals/siamcomp/Ambainis07}]
%Let $N = 2^n$ and $H~: \zo^n \mapsto \zo^n$ that is either injective, or admit a single collision pair: $x,y, H(x) = H(y)$. There exists a quantum algorithm that, given superposition access to $H$, finds which is the case, and finds the collision pair. It runs in time $\bigOt{ N^{2/3} }$ using $\bigO{ N^{2/3}}$ qubits (in the qRAM model) and $\bigO{N^{2/3}}$ superposition queries to $H$.
%\end{theorem}

Using this framework generically, we can recover the complexity of Ambainis' algorithm for collision search: $\bigOt{2^{m/3}}$ for any codomain bit-size $m$. We use the Johnson graph $J(2^n, 2^{m/3})$. Its spectral gap is approximately $2^{-m/3}$. A vertex is marked if and only if it contains a collision, so the probability of being marked is approximately $2^{2m/3-m} = 2^{-m/3}$. Using a quantum data structure for unordered sets, we can implement the Setup operation in time $\bigOt{2^{m/3}}$, the Update and the Check in $\mathsf{poly}(n)$. The formula of~\Cref{thm:qwalk} gives the complexity $\bigOt{2^{m/3}}$.

\subsection{Details of the MNRS Framework}

In the $d$-regular graph $G = (V,E)$, for each $x \in V$, let $N_x$ be the set of neighbors of $x$, of size $d$. In the case $G = J(N,R)$, we have $d = R(N-R)$. For a vertex $x$, let $\ket{x}$ be an arbitrary encoding of $x$ as a quantum state, let $D(x)$ be a \emph{data structure} associated to $x$, and let $\ket{\xh} = \ket{x} \ket{D(x)}$.

%The only required property is that: $\forall x \neq y, \langle \ket{x} | \yh \rangle = 0$.

\begin{remark}
The encoding of $x$ is commonly thought of as the set itself, and the data structure as the images of the set by $f$. But whenever we look at quantum walks from the perspective of time complexity (and not query complexity), an efficient quantum data structure is already required for $x$ itself, i.e., an unordered set data structure in the case of a Johnson graph, and one cannot really separate $x$ from $D(x)$. This is why we will favor the notation $\ket{\xh}$.
\end{remark}

For a vertex $x$, let $\ket{p_x}$ be the uniform superposition over its neighbors: $\ket{p_x} = \frac{1}{\sqrt{d}} \sum_{y \in N_x} \ket{y}$, and: $\ket{\widehat{p_x}} = \frac{1}{\sqrt{d}} \sum_{y \in N_x} \ket{\yh}$. From now on, we consider a walk on \emph{edges} rather than vertices in the graph, and introduce:
\begin{align*}
\begin{cases}
\ket{\psi_U} = \frac{1}{\sqrt{|V|}} \sum_{x \in V} \ket{\widehat{x}} \ket{p_x} \text{ the superposition of vertices (and neighbors)} \\
\ket{\psi_M} = \frac{1}{\sqrt{|M|}} \sum_{x \in M} \ket{\widehat{x}} \ket{p_x} \text{ the superposition of marked vertices} \\
A = \mathrm{span} \{\ket{\xh}\ket{p_x}\}_{x \in V} \\
B = \mathrm{span} \{\ket{\pyh}\ket{y}\}_{y \in V} \\
\end{cases}
\end{align*}

Let $\mathrm{Ref}_A$ and $\mathrm{Ref}_B$ be respectively the reflection over the space $A$ and the space $B$. The core of the MNRS framework is to use these operations to emulate a reflection over $\ket{\psi_U}$. By alternating such reflections with reflections over $\ket{\psi_M}$ (using the checking procedure), the quantum walk behaves exactly as a quantum search, and the analysis of~\Cref{subsection:grover} applies.

\begin{proposition}[From~\cite{DBLP:journals/siamcomp/MagniezNRS11}] \label{Proposiiton:MNRS_Main}
	Let $W = \mathrm{Ref}_B \mathrm{Ref}_A$. We have $
	\bra{\psi_{U}} W\ket{\psi_{U}} = 1$. 
	For any other eigenvector $\ket{\psi}$ of $W$, we have 
	$\bra{\psi} W\ket{\psi}  = e^{i\theta}$
	with $\theta  \in [2\sqrt{\delta},\pi/2]$.
\end{proposition}

To reflect over $\ket{\psi_U}$, we perform a \emph{phase estimation} of the unitary $W$, which allows to separate the part with eigenvalue 1, from the part with eigenvalue $e^{i\theta}$ with $\theta \in [2\sqrt{\delta},\pi/2]$. The phase estimation circuit needs to call $W$ a total of $\bigO{\frac{1}{\sqrt{\delta}}}$ times to estimate $\theta$ up to sufficient precision. It has some error, which can be made insignificant with a polynomial increase in complexity; thus in the following, we will consider the reflection $\mathrm{Ref}_{U}$ to be exact.

To construct $W$, we need to implement $\mathrm{Ref}_A$ and $\mathrm{Ref}_B$. We first remark that:
\begin{equation}
\mathrm{Ref}_B = \SWUP \circ \mathrm{Ref}_A \circ \SWUP \enspace ,
\end{equation}
where $\SWUP \ket{\xh}\ket{y} = \ket{\yh}\ket{x}$. This $\SWUP$ (Swap-Update) operation can furthermore be decomposed into an update of the database ($\UP_D$) followed by a register swap:
\begin{equation}
\ket{\xh}\ket{y} = \ket{x}\ket{D(x)}\ket{y} \xrightarrow{\UP_D} \ket{x}\ket{D(y)}\ket{y} \xrightarrow{\Swap} \ket{y}\ket{D(y)}\ket{x} = \ket{\yh}\ket{x} \enspace ,
\end{equation}
so $\SWUP = \Swap \circ \UP_D$. 

We would then implement $\mathrm{Ref}_A$ using an update unitary that, from a vertex $x$, constructs the uniform superposition of neighbors. However this would require us to write $\log_2(|V|)$ data, and in practice, $|V|$ is doubly exponential (the vertex is represented by an exponential number of bits).
%This $UP_D$ operator is what requires queries and, in the imperfect update case, which will be the source of errors when computing $W$. 
Thankfully, in $d$-regular graphs, and in particular in Johnson graphs, we can avoid this loophole by making the encoding of edges more compact. Instead of storing a pair of vertices $(x,y)$, which will eventually result in having to rewrite entire vertices, we can store a single vertex and a \emph{direction}, or \emph{coin}.

\subsection{Vertex-coin Encoding}

The encoding is a reversible operation: $O_{\Enc} \ket{\xh}\ket{y} = \ket{\xh}\ket{c_{x \rightarrow y}} \enspace$,
which compresses an edge $(x,y)$ by replacing $y$ by a much smaller register of size $\lceil \log_2(d) \rceil$. Note that we only need the \emph{existence} of such a circuit. We never use it during the algorithms; all operations are directly performed using the compact encoding.

Let $\ket{\psi_{\Unif}^{coin}} = \frac{1}{\sqrt{d}} \sum_{c} \ket{c}$ be the uniform superposition of coins. In the vertex-coin encoding, $\mathrm{Ref}_A$ corresponds to $I \otimes Ref_{\ket{\psi_{\Unif}^{coin}}}$:
$$ \mathrm{Ref}_A = O_{\Enc}^{-1} \circ \left(I \otimes \mathrm{Ref}_{\ket{\psi_{\Unif}^{coin}}}\right) \circ O_{\Enc}.$$
Now, for the $\SWUP$ operation, we have to decompose again $\UP_D$ and $\Swap$ in the encoded space. First, we define $\UP'_D$ such that: 
$$\ket{x}\ket{D(x)}\ket{c_{x \rightarrow y}} \xrightarrow{UP'_{D}} \ket{x}\ket{D(y)}\ket{c_{x \rightarrow y}}.$$

Moreover, we define $\Swap'$ such that:
$$\ket{x}\ket{c_{x \rightarrow y}} \xrightarrow{\Swap'} \ket{y}\ket{c_{y \rightarrow x}}.$$
and we define $\SWUP' = \Swap' \circ \UP'_D$ (we abuse notation here, by extending $\Swap'$ where we apply the identity to the middle register), so:
\begin{equation*}
\SWUP' \ket{\xh}\ket{c_{x \rightarrow y}} = \ket{\yh}\ket{c_{y \rightarrow x}} \enspace,
\end{equation*}
and $\SWUP' = O_{\Enc} \circ \SWUP \circ O_{\Enc}^{-1}$. So we define 
\begin{equation}
\begin{cases}
	\mathrm{Ref}'_A = I \otimes \mathrm{Ref}_{\ket{\psi_{\Unif}^{coin}}} = O_{\Enc} \circ Ref_A \circ O_{\Enc}^{-1} \\
	\mathrm{Ref}'_B = \SWUP' \circ \mathrm{Ref}'_A \circ \SWUP' = O_{\Enc} \circ \mathrm{Ref}_B \circ O_{\Enc}^{-1}  \\
	W' = \mathrm{Ref}'_B \circ \mathrm{Ref}'_A 
	\end{cases}
\end{equation}
By putting everything together, we have $W' = O_{\Enc} \circ W \circ O_{\Enc}^{-1}$. So we can use Proposition~\ref{Proposiiton:MNRS_Main} to have the spectral properties and perform phase estimation on $W'$, and combine afterwards with~\autoref{Proposition:GroverOriginal}. Since constructing the uniform superposition of coins is trivial, all relies on the unitary $\SWUP'{}$.

\begin{theorem}[MNRS, adapted]\label{thm-mnrs-adapted}
Let $\ket{\xh}$ be an encoding of the vertex $x$ (incl. data structure) and assume that a vertex-coin encoding is given. Let $\alpha = \arcsin \sqrt{\epsilon}$. Starting from the state: $\frac{1}{\sqrt{|V|}} \sum_{x \in V} \ket{\xh} \ket{\psi_{\Unif}^{coin}}$, applying $\floor{ \frac{\pi/2 -\alpha}{2\alpha}}$ iterates of: $\bullet$~a checking procedure which flips the phase of marked vertices; $\bullet$~a phase estimation of $W'$; then apply the checking again and measure. With probability at least $1-4\alpha^2$, we measure 1 and collapse on the uniform superposition of marked vertices.
\end{theorem}

%\YS{since we've changed proposition 3, we need to change this theorem statement as well}

\paragraph{Coins for a Johnson Graph.}
In a Johnson graph $J(N,R)$, a coin $c = (j,z)$ is a pair where:
\begin{itemize}
\item $j \in [R]$ is the index of the element that will be removed from the current vertex (given an arbitrary ordering, e.g. the lexicographic ordering of bit-strings).
\item $z \in [N-R]$ is the index of an element that does not belong to the current vertex, and will be added as a replacement.
\end{itemize}

\newcommand{\ALLOC}{\mathrm{ALLOC}}
\newcommand{\INSERT}{\mathrm{INSERT}}
\newcommand{\LOOKUP}{\mathrm{LOOKUP}}
\newcommand{\QLOOKUP}{\mathrm{QLOOKUP}}

\newcommand{\CHECK}{\mathrm{CHECK}}
\newcommand{\EXTRACT}{\mathrm{EXTRACT}}

%\input{./files/ideal-walk.tex}

% !TeX root = ../collision-walk.tex

\section{A Chained Quantum Walk to Find Many Collisions}
\label{sec:many_collisions}

In this section, we prove our main result.

\begin{theorem}\label{thm:main-result}
Let $f : \zo^n \to \zo^m$, $n \leq m \leq 2n$ be a random function. Let $k \leq \min(2n-m, m/4)$. There exists a quantum algorithm making $\bigO{2^{2k/3 + m/3}}$ quantum queries to $f$ and using $\bigOt{ 2^{2k/3 + m/3} }$ Clifford+T+\textsf{qRAMR}+\textsf{qRAMW} gates, that outputs $2^k$ collision pairs of $f$.
\end{theorem}

Our new algorithm, which is detailed in~\autoref{sec:new-algo} and~\autoref{sec:complexity}, solves the case $k \leq \frac{m}{4}$. The case $k \leq 2n-m$ was already solved by adapting the BHT algorithm, as detailed in~\autoref{subsection:general-case}.

\subsection{New Algorithm}\label{sec:new-algo}

We detail here our \emph{chained quantum walk} algorithm. We start by introducing some necessary notation.

Recall that the Johnson graph $J(N,R)$ is the regular, undirected graph whose vertices are subsets of size $R$ of $[N]$, and edges connect each pair of vertices which differ in exactly one element. We identify $[N]$ with $\zo^n$, the domain of $f$.

We assume that an efficient quantum unordered set data structure is given, which makes vertices in the Johnson graph correspond to quantum states, while allowing to implement efficiently the operations required for the MNRS quantum walks. It will be detailed in~\autoref{sec:data}. In the following we write $\ket{S}$ for the quantum state corresponding to a set $S$.

Assume that a table of (multi)-collisions of $f$ is given, which we denote $C$. This table contains entries of the form: $u : (x_1, \ldots, x_r)$ where $f(x_1) = \ldots = f(x_r) = u$ forms a multicollision of $f$, indexed by the image. We define the \emph{size} of $C$, its set of \emph{preimages} and its set of \emph{images}:
\begin{equation}
\begin{cases}
\mathsf{Preim}(C) := \bigcup_{ u : (x_1, \ldots, x_r) \in C } \{ x_1, \ldots, x_r \} \\
\mathsf{Im}(C) := \bigcup_{ u : (x_1, \ldots, x_r) \in C } \{ u \}
\end{cases}
\end{equation}

Given a table $C$, given a size parameter $R$, we define the two sets of sets:
\begin{align*}
V_R^C &:= \{ S \subseteq \left(\zo^n \backslash \mathsf{Preim}(C)\right), |S| = R \} \\
M_R^C &:= \{ S \subseteq \left(\zo^n \backslash \mathsf{Preim}(C)\right), |S| = R , \\
&\qquad \qquad	\left( \exists x \neq y \in S, f(x) = f(y) \vee \exists x \in S, f(x) \in \mathsf{Im}(C) \right) \}
\end{align*}

\paragraph{Idea of Our Algorithm.}
After running a quantum walk on a Johnson graph, we obtain a superposition of vertices which contain a collision. Our main idea is that, after removing this collision from the vertex and measuring it, it collapses to a superposition \emph{close} to a uniform superposition of vertices of smaller size. We can then restart a quantum walk on this smaller Johnson graph, which runs similarly as the previous one.% We repeat this process an exponential number of times.

The definition of $V_R^C$ and $M_R^C$ allows to formalize this idea: the first one will be the set of vertices for the current walk, and the second one its set of \emph{marked} vertices. As we can see, the current walk excludes a set of previously measured inputs, and a vertex is marked if it leads to a new collision, or to a preimage of one of the previously measured images. The second case simply extends one of the currently known multicollision tuples. The probability for a vertex to be marked can be easily computed, and we just need to bound it as follows:
\[ \max \left(\frac{R |\mathsf{Im}(C)|}{2^m} , \frac{ R(R-1) }{ 2^m}  \right)  \leq  \varepsilon_{R, C} \leq \frac{R |\mathsf{Im}(C)|}{2^m} +  \frac{ R(R-1) }{ 2^m} \enspace, \]
since any vertex containing a collision, or a preimage from the table $C$, is marked.
% In practice, the case of containing a collision will remain more common than the case of containing a preimage.

In~\autoref{sec:data}, we will show that with an appropriate data structure, there exists an \emph{extraction} algorithm $\EXTRACT{}$ which does the following:
\begin{equation*}
\EXTRACT{} : C, R, \frac{1}{\sqrt{|M_R^C|}} \sum_{ S \in M_R^C} \ket{S} \mapsto C', R', \frac{1}{\sqrt{|  V_{R'}^{C'} \backslash M_{R'}^{C'}|}} \sum_{S \in V_{R'}^{C'} \backslash M_{R'}^{C'}} \ket{S} \enspace,
\end{equation*}
where $R' = R-r$ for some value $r$, and $C'$ contains exactly $r$ new elements (collisions adding new entries, or preimages going into previous entries). Thus, $\EXTRACT{}$ transforms the output of a successful walk into the set of \emph{unmarked vertices} for the next walk.

We can now give~\autoref{algo:full-algo}, depending on a tunable parameter $\ell$.

\begin{algorithm}[htb]
	\DontPrintSemicolon
	\KwIn{quantum access to $f~: \zo^n \to \zo^m$, parameter $k$}
	\KwOut{a table of multicollisions $C$ such that $|\mathsf{Im}(C)| \geq 2^k$}
	$C \leftarrow \emptyset$, $R \leftarrow 2^\ell$ \tcc{Initialize an empty table}
	$\ket{\psi} \leftarrow \sum_{S \in V_{2^\ell}^C} \ket{S}$ \;
	\While{$|\mathsf{Im}(C)| < 2^k$}{
		Run the quantum walk:
			\begin{itemize}
			\item Starting state: $\ket{\psi} = \sum_{S \in V_R^C \backslash M_R^C} \ket{S}$
			\item Graph: $J(\zo^n \backslash \mathsf{Preim}(C), R)$ (Johnson graph with vertices of size $R$, \\
				excluding the preimages of $C$)
			\item Marked vertices: $M_R^C$
			\item Iterates: $\floor{(\pi/2 )/(2\alpha)}$, where $\alpha = \arcsin \sqrt{\varepsilon_{R, C}}$
			\item Spectral gap: $\delta \simeq \frac{1}{R}$
			\end{itemize}
		Apply $\CHECK{}$ and measure the result: let \textsf{flag} be the output\;
		\uIf{\textsf{flag} is true}{
			\tcc{The state collapses on: $\sum_{S \in M_R^C} \ket{S}$}
			Apply $\EXTRACT{}$ (contains measurements)
			\begin{itemize}
			\item Update the table $C$
			\item Update the current width $R$
			\item Update the state: $\ket{\psi} = \sum_{S \in V_R^C \backslash M_R^C} \ket{S}$
			\end{itemize}
		}
		\tcc{Otherwise, the state collapses on: $\sum_{S \in V_R^C \backslash M_R^C} \ket{S}$ for the previous $R$ and $C$. There is nothing to extract from it, $C$ and $R$ remain unchanged.}
	}
	\Return{$C$}
	\caption{Chained quantum walk algorithm for multiple collisions.}\label{algo:full-algo}
\end{algorithm}

%In order to transform the state into a superposition of unmarked vertices for the next walk, we extract, and we end up constructing a list of multicollisions rather than collisions. Fortunately, there exists a polynomial $p(n)$ such that with overwhelming probability, the function $f$ does not have a $p(n)+1$-collision. So there is always a polynomial relation between the size of $U$ (the set of forbidden images) and $C$ (the multicollision tuples that we found). This greatly simplifies the analysis.

\subsection{Complexity Analysis}\label{sec:complexity}

%The complexity of~\autoref{algo:full-algo} follows quite easily from~\autoref{thm-mnrs-adapted} (a quantum walk starting from the superposition of \emph{unmarked vertices}) and the correctness of $\EXTRACT{}$.

%In our analysis of~\autoref{algo:full-algo}, we prove the following.

\begin{theorem}[Time-memory tradeoff]\label{thm:walk-tradeoff}
For all $k \leq \ell \leq \min(2k/3+m/3,m/2)$, \autoref{algo:full-algo} computes $2^k$ collisions using $\bigOt{2^\ell}$ qubits and $\bigOt{2^{k+m/2-\ell/2}}$ Clifford+T+\textsf{qRAMR}+\textsf{qRAMW} gates.
\end{theorem}

\begin{proof}
It should be noted that~\autoref{algo:full-algo} outputs a set of multicollisions rather than collisions. But for a random function with a domain of $n$ bits, there is a polynomial (in $n$) limit to the width of multicollisions that can appear. Thus, we have a polynomial relation $p(n)$ between $|\mathsf{Preim}(C)|$ and $|\mathsf{Im}(C)|$. In particular, by taking $2^\ell$ greater than $p(n) 2^{k+1}$, we ensure that during the algorithm, $R > 2^{\ell-1}$. In particular, we never run out of elements.

Secondly, we can bound $\varepsilon_{R,C} \geq \frac{R(R-1)}{2^m}$. This allows to upper bound easily the time complexity of any of the walks: if the current vertex size is $R$ then it runs for $\bigO{ 2^{m/2} / R}$ iterates, and each iterate contains $\bigOt{\sqrt{R}}$ operations. The constants in the $\mathcal{O}$ are the same throughout the algorithm. This means that we can upper bound the complexity of each walk by $\bigOt{ 2^{m/2} / \sqrt{R} } \leq \bigOt{ 2^{m/2 - \ell / 2}}$.

By~\autoref{thm-mnrs-adapted}, the success probability of this walk is bigger than $1 - 4 \varepsilon_{R,C}$. If we do not succeed, the $\CHECK{}$ followed by a measurement make the current state collapse again on the superposition of unmarked vertices, and we run the exact same walk again. Note that for this algorithm to work, we must have $\varepsilon_{R,C} < 0.5$. This corresponds to the probability that an element in the list collides with another element (either in the list itself or in the set of forbidden preimages), which is a $\bigOt{2^{2\ell-m}}$. Hence, we must have $\ell \leq m/2$.

Then, as $\ell \leq 2k/3+m/3$, the final complexity of the algorithm is
\[ \bigOt{ 2^\ell + 2^k 2^{m/2 - \ell / 2}  } = \bigOt{2^{k+m/2-\ell/2}} \enspace. \qquad \ensuremath{\Box} \]
\end{proof}

\section{Quantum Radix Trees and Extractions}\label{sec:data}

In this section, we detail the \emph{quantum radix tree} data structure, a history-independent unordered set data structure introduced in~\cite{DBLP:phd/basesearch/Jeffery14}. We show that it allows to perform, exactly and in a polynomial number of Clifford+T+\textsf{qRAMR} + \textsf{qRAMW} gates, the two main operations required for our walk: $\SWUP'$ and $\EXTRACT{}$. We describe these operations in pseudocode, while ensuring that they are reversible and polynomial.

%We show that it allows to perform the operation $\SWUP'$ exactly and in polynomial time. We also show the additional property that collisions can be extracted from the set, in a way that allows to reuse the remaining elements as the starting point of another walk.

%The description of the quantum radix trees and their operations on the Clifford+T+\textsf{qRAMR}+\textsf{qRAMW} gate set would be a challenging task. Instead, we describe them at a pseudocode level. For an asymptotic analysis, we just have to ensure that the operations performed are reversible and that their complexity is polynomial.

\subsection{Logical Level}

Following~\cite{DBLP:phd/basesearch/Jeffery14}, the \emph{quantum radix tree} is an implementation of a radix tree storing an unordered set $S$ of $n$-bit strings. It has one additional property: its concrete memory layout is history-independent. Indeed, there are many ways to encode a radix tree in memory, and as elements are inserted and removed, we cannot have a unique bit-string $T(S)$ representing a set $S$. We use instead a uniform superposition of all memory layouts of the tree, which makes the \emph{quantum state} $\ket{T(S)}$ unique, and independent of the order in which the elements were inserted or removed. Only the entry point (the root) has a fixed position.

%The \emph{quantum radix tree} encodes an unordered set $S$ of $n$-bit strings of size $R$, which can be used as identifiers for more complex elements. This encoding is a one-to-one mapping from such sets $S$ to a set of orthogonal quantum states $\ket{T(S)}$, which we will specify in what follows. 

We separate the encoding of $S$ into $\ket{T(S)}$ in two levels: first, a \emph{logical level}, in which $S$ is encoded as a unique radix tree $R(S)$; second, a \emph{physical level}, in which $R(S)$ is encoded into a quantum state $\ket{T(S)}$. The logical mapping $S \to R(S)$ is standard.

\begin{definition}[From~\cite{DBLP:phd/basesearch/Jeffery14}]
Let $S$ be a set of $n$-bit strings. The radix tree $R(S)$ is a binary tree in which each leaf is labeled with an element of $S$, and each edge with a substring, so that the concatenation of all substrings on the path from the root to the leaf yields the corresponding element. Furthermore, the labels of two children of any non-leaf node start with different bits.
\end{definition}

By convention, we put the ``0'' bit on the left, and ``1'' on the right. In addition to the $n$-bit strings stored by the tree, we append to each node the value of an $\ell$-bit \emph{invariant} which can be computed from its children, and depends only on the logical structure of the radix tree, not its physical structure. Typically the invariant can count the number of elements in the tree.

\begin{figure}[tb]
\centering
\begin{tikzpicture}[level distance=8mm, nodes={draw,rectangle,rounded corners=.2cm,->}, level 1/.style={sibling distance=40mm}, level 2/.style={sibling distance=15mm}, level 3/.style={sibling distance=15mm}]
\node[draw, circle] { }
    child { node[draw, circle] {} 
    		child{ node {$\mathtt{0000}$} edge from parent node[draw=none, left] {\texttt{00}} }
    		child{ node {$\mathtt{0010}$} edge from parent node[draw=none, left] {\texttt{10}} }
    		edge from parent node[draw=none, left] {\texttt{00}}
    }
    child { node[draw, circle] {}
    		child { node[draw, circle] {} 
    			child{ node {$\mathtt{1001}$} edge from parent node[draw=none, left] {\texttt{01}} }
    			child{ node {$\mathtt{1011}$} edge from parent node[draw=none, left] {\texttt{11}} }
    			edge from parent node[draw=none, left] {\texttt{0}}
    		}
    		child { node {$\mathtt{1111}$} edge from parent node[draw=none, left] {\texttt{111}}}
    		edge from parent node[draw=none, left] {\texttt{1}}
    	};
\end{tikzpicture}
\caption{Tree $R(S)$ representing the set $S = \{\mathtt{0000},\mathtt{0010}, \mathtt{1001}, \mathtt{1011}, \mathtt{1111}\}$ (the example is taken from~\cite{DBLP:phd/basesearch/Jeffery14}).}\label{fig:tree-logical}
\end{figure}
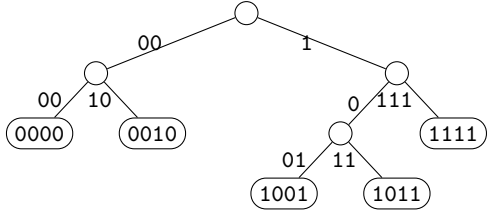

\iffalse
At the logical level, the operations of \emph{lookup}, \emph{insertion} and \emph{deletion} can be implemented in polynomial time as follows:
\begin{itemize}
\item Lookup: given a string $x$ and $R(S)$, we want to know if $x \in S$. From the root, we traverse the tree by choosing edges such that the concatenated labels form a prefix of $x$. We either find a leaf labeled by $x$, or stop with no valid outgoing edge.
\item Insertion of $x$: we traverse the tree similarly. Assuming that the set does not contain $x$, when the traversal stop, we arrive at a node such that its two outgoing edges do not form prefixes of $x$. However, one of them will share a non-empty prefix with $x$. We split this edge into two edges and insert a new node between them. We insert the new element as a leaf of this new node. The edge labels are updated accordingly.
\item Deletion: assuming that the element to remove belongs to $S$, we uncompute the insertion procedure.
\end{itemize}
\fi

\subsection{Memory Representation}

We now detail the correspondence from $R(S)$ to $\ket{T(S)}$. We suppose that a quantum \emph{bit-string data structure} is given, that handles bit-strings of length between 0 and $n$ and performs operations such as concatenation, computing shared prefixes, testing if the bit-string has a given prefix, in time $\mathsf{poly}(n)$.

\paragraph{State of the Memory.}
We suppose that $\bigO{M n}$ qubits of memory are given, where $M \geq R$ will be set later on. We divide these qubits into $M$ \emph{cells} of $\bigO{n}$ qubits each, which we index from $0$ to $M-1$. We encode cell addresses on $m = \ceil{\log_2 M} + 1$ bits, and we also define an ``empty'' address $\bot$. Each cell will be either empty, or contain a node of the radix tree, encoded as a tuple $(i, a_l, a_r, \ell_l, \ell_r)$ where: 
%$\bullet$~$i$ is the value of the invariant; $\bullet$~$a_l$ and $a_r$ ($m$-bit strings) are respectively the memory addresses of the cells holding the left and right children, either valid indices or $\bot$; $\bullet$~$\ell_l$ and $\ell_r$ are the labels of the left and right edges. If the node is a leaf, $a_l = a_r = \bot$ and $\ell_l = \ell_r = \varepsilon$ (empty string).
\begin{itemize}
\item $i$ is the value of the invariant
\item $a_l$ and $a_r$ ($m$-bit strings) are respectively the memory addresses of the cells holding the left and right children, either valid indices or $\bot$. A node with $a_l = a_r = \bot$ is a leaf.
\item $\ell_l$ and $\ell_r$ are the labels of the left and right edges. ($\varepsilon$ if the node is a leaf, where $\varepsilon$ is the empty string).
\end{itemize}

In other words, we have added to the tree $R(S)$ a choice of memory locations for the nodes, which we name informally the \emph{memory layout} of the tree. The structure of $R(S)$ itself remains independent on its memory layout. 
%We can observe that the tree is always binary, and if a node is a leaf, then the bit-string that it represents is obtained by concatenating all labels on the path from it to the root.

The root of the tree is stored in cell number 0. In~\autoref{fig:tree-physical}, we give an example of a memory representation of the tree $R(S)$ of~\autoref{fig:tree-logical}. We take as invariant the number of leaves which, at the root, gives the number of elements in the set. 

\begin{figure}[tb]
\centering
\begin{tikzpicture}[level distance=10mm, nodes={draw,rectangle,rounded corners=.2cm,->}, level 1/.style={sibling distance=50mm}, level 2/.style={sibling distance=23mm}, level 3/.style={sibling distance=23mm}]

\node at (0,0) { $(5, 1, 2, \mathtt{00}, \mathtt{1})$ }
    child { node { $(2, 3, 4, \mathtt{00}, \mathtt{10})$ } 
    		child{ node { $(1, \bot, \bot, \varepsilon, \varepsilon)$ }  }
    		child{ node { $(1, \bot, \bot, \varepsilon, \varepsilon)$ }  }
    }
    child { node {$(3, 5, 7, \mathtt{0}, \mathtt{111})$}
    		child { node {$(2, 8, 9, \mathtt{01}, \mathtt{11})$} 
    			child{ node {$(1, \bot, \bot, \varepsilon, \varepsilon)$}  }
    			child{ node {$(1, \bot, \bot, \varepsilon, \varepsilon)$}  }
    		}
    		child { node {$(1, \bot, \bot, \varepsilon, \varepsilon)$} }
    	};
 \node[draw=none] at (-6,-1.5) {Structured view:};
\end{tikzpicture}

\medskip
\begin{tabular}{lc|c|c|c|c|c|c|c}
&0 & $(5, 1, 2, \mathtt{00}, \mathtt{1})$ & 1 & $(2, 3, 4, \mathtt{00}, \mathtt{10})$ & 2 & $(3, 5, 7, \mathtt{10}, \mathtt{11})$ & 3 & $(1, \bot, \bot, \varepsilon, \varepsilon)$ \\
Actual memory content: &4 & $(1, \bot, \bot, \varepsilon, \varepsilon)$ & 5 & $(2, 8, 9, \mathtt{01}, \mathtt{11})$ & 6& Empty cell & 7 & $(1, \bot, \bot, \varepsilon, \varepsilon)$ \\
&8 & $(1, \bot, \bot, \varepsilon, \varepsilon)$ & 9 & $(1, \bot, \bot, \varepsilon, \varepsilon)$ \\
\end{tabular}

\caption{Example of memory layout for the tree of~\autoref{fig:tree-logical}, holding the set $S = \{\mathtt{0000},\mathtt{0010}, \mathtt{1001}, \mathtt{1011}, \mathtt{1111}\}$.}\label{fig:tree-physical}
\end{figure}
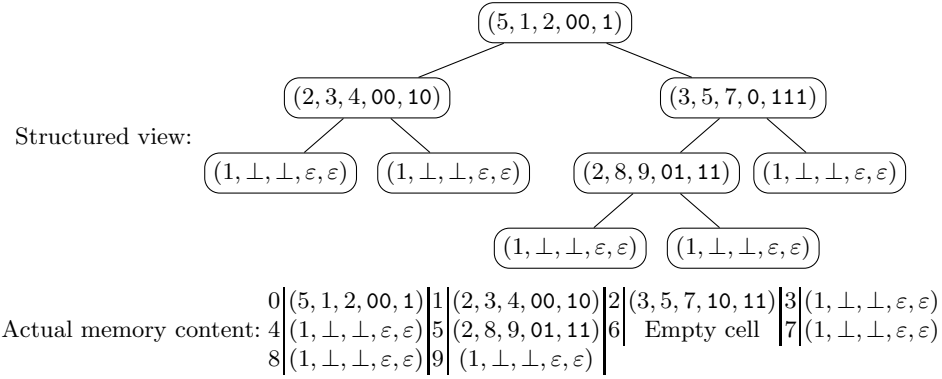

A radix tree encoding a set of size $R$ contains $2 R-1$ nodes (including the root), which means that we need (a priori) no more than $M=2 R-1$ cells in our memory. In addition to the bit-strings $x$, we could add any data $d_x$ to which $x$ serves as a unique index. (This means adding another register which is non-empty for leaf nodes only). Finally, it is possible to account for multiplicity of elements in the tree by adding multiplicity counters, but since this is unnecessary for our applications, we will stick to the case of unique indices.

\paragraph{Definition.}
Let $S$ be a set of size $R$, encoded in a radix tree with $2R-1$ nodes. We can always take an arbitrary ordering of the nodes in the tree, for example the lexicographic ordering of the paths to the root (left = 0, right = 1). This means that, for any sequence of non-repeating cell addresses $I$, of length $2R-1$, we can define a mapping: $S, I \mapsto T_I(S)$ which specifies the writing of the tree in memory, by choosing the addresses $I = (i_1 = 0, \ldots, i_{2R-1})$ for the elements. For example, the tree of~\autoref{fig:tree-physical} would correspond to the sequence $(0,1,3,4,2,5,8,9,7)$. We can then define the \emph{quantum radix tree encoding $S$} as the quantum state:
\begin{equation}
\ket{ T(S) } = \sum_{\text{valid sequences } I} \ket{T_I(S)} \enspace,
\end{equation}
where we take a uniform superposition over all valid memory layouts.

For two different sets $S$ and $S'$, and for any pair $I,I'$ (even if $I' = I$), we have $T_{I'}(S) \neq T_I(S')$: the encodings always differ. This means that, as expected, we have $\braket{T(S)}{T(S')} = 0$.
%\begin{itemize}
%\item In the case of a singleton, we only have the root in its fixed position, so the superposition contains a single basis vector.
%\item For two different sets $S$ and $S'$, and for any pair $I,I'$ (even if $I' = I$), we have $T_{I'}(S) \neq T_I(S')$: the encodings always differ. This means that, as expected, we have $\braket{T(S)}{T(S')} = 0$.
%\end{itemize}

\paragraph{Memory Allocator.}
In order to maintain this uniform superposition over all possible memory layouts, we need an implementation of a \emph{memory allocator}. This unitary $\ALLOC{}$ takes as input the current state of the memory, and returns a uniform superposition over the indices of all currently unoccupied cells. Possible implementations of $\ALLOC{}$ are detailed in~\autoref{subsection:memory-alloc}.

\subsection{Basic Operations}

\iffalse
At the logical level, the operations of \emph{lookup}, \emph{insertion} and \emph{deletion} can be implemented in polynomial time as follows:
\begin{itemize}
\item Lookup: given a string $x$ and $R(S)$, we want to know if $x \in S$. From the root, we traverse the tree by choosing edges such that the concatenated labels form a prefix of $x$. We either find a leaf labeled by $x$, or stop with no valid outgoing edge.
\item Insertion of $x$: we traverse the tree similarly. Assuming that the set does not contain $x$, when the traversal stop, we arrive at a node such that its two outgoing edges do not form prefixes of $x$. However, one of them will share a non-empty prefix with $x$. We split this edge into two edges and insert a new node between them. We insert the new element as a leaf of this new node. The edge labels are updated accordingly.
\item Deletion: assuming that the element to remove belongs to $S$, we uncompute the insertion procedure.
\end{itemize}
\fi

We show how to operate on the quantum radix trees in $\mathsf{poly}(n)$ Clifford+T +\textsf{qRAMR}+\textsf{qRAMW} gates. We start with the basics: lookup, insertion and deletion.

\paragraph{Lookup.}
We define a unitary $\LOOKUP{}$ which, given $S$ and a new element $x$, returns whether $x$ belongs to $S$:
\begin{equation}
\LOOKUP~: \ket{x} \ket{T(S)} \ket{0} \mapsto \ket{x} \ket{T(S)} \ket{x \in S} \enspace.
\end{equation}
% this works in the same way for all memory layouts of $R(S)$. The only change is in the addresses of the cells that we need to query.
We implement $\LOOKUP{}$ by descending in the radix tree $R(S)$; he pseudocode is given in~\autoref{algo:lookup}. Since the ``while'' loop contains at most $n$ iterates, quantumly these $n$ iterates are performed controlled on a flag that says whether the computation already ended. After obtaining the result, they are recomputed to erase the intermediate registers.

\begin{algorithm}[tb]
	\DontPrintSemicolon
	\KwIn{element $x$, quantum radix tree $T(S)$}
	\KwOut{whether $x \in S$}
	$(i, a_l, a_r, \ell_l, \ell_r) \gets$ root  \;
	$y \gets \varepsilon$ (empty string) \;
	\While{ $a_l \neq \bot$ (node is not a leaf) }{
		\uIf{$y || \ell_l$ is a prefix of $x$}{
			$y \gets y || a_l$ \;
			$(i, a_l, a_r, \ell_l, \ell_r) \gets$ node at address $a_l$\;
		}
		\uElseIf{$y || \ell_r$ is a prefix of $x$}{
			$y \gets y || a_l$ \;
			$(i, a_l, a_r, \ell_l, \ell_r) \gets$ node at address $a_r$\;
		}
		\Else{
			Break (not a solution) \;
		}
	}
	\Return{true if $y = x$}
	\caption{$\LOOKUP{}$ as a classical algorithm.}\label{algo:lookup}
\end{algorithm}

%\YS{dans l'algo 2 pourquoi ne pas output true quand $y=x$?}\XB{}

\paragraph{Insertion.}
We define a unitary $\INSERT{}$, which, given a new element $x$, inserts $x$ in the set $S$. If $x$ already belongs to $S$, its behavior is unspecified.
\begin{equation}
\INSERT ~: \ket{x} \ket{T(S)} \mapsto \ket{x} \ket{T(S\cup \{x\})} \enspace.
\end{equation}

The implementation of $\INSERT{}$ is more complex, but the operation is still reversible. The pseudocode is given in~\autoref{algo:insert}. At first, we find the point of insertion in the tree, then we call $\ALLOC{}$ twice to obtain new memory addresses for two new nodes. We modify locally the layout to insert these new nodes, including a new leaf for the new element $x$. Then, we update the invariant on the path to the new leaf. Finally, we uncompute the path to the new leaf (all the addresses of the nodes on this path). To do so, we perform a loop similar to $\LOOKUP{}$, given the knowledge of the newly inserted element $x$. 

\begin{algorithm}[htbp]
	\DontPrintSemicolon
	\KwIn{element $x$, quantum radix tree $T(S)$}
	\KwOut{element $x$, quantum radix tree $T(S \cup \{x\})$}
	Find the first node $j_1: (i,a_l, a_r, \ell_l, \ell_r)$ such that $y$ is a prefix of $x$, $y || \ell_l$ is not a prefix of $x$ and $y || \ell_r$ is not a prefix of $x$ either.  Write all the addresses of the nodes on the path from the root to $j_1$ \;
	\tcc{If at this point we have found that the element belongs to $S$ instead, then the rest of the computation is meaningless.}
	\tcc{By construction $\ell_l$ starts with 0 and $\ell_r$ starts with 1. One of them shares a non-empty prefix $z$ with the remaining part of $x$. Without loss of generality, we assume that it is $\ell_l$.}
	Let $\ell_l = z || t$ and $x = y || z || x'$ \;
	Call $\ALLOC{}$ to obtain an address $j_2$\;
	Replace $a_l$ with $j_2$ in the node $j_1: (i,a_l, a_r, \ell_l, \ell_r)$ (move $a_l$ to a temporary register) \;
	Call $\ALLOC{}$ to obtain an address $j_3$\;
	Write at address $j_3$: $(*, \bot, \bot, \varepsilon, \varepsilon)$\;
	\tcc{Information at this point: $x, a_l, j_2, j_3$, the path to $j_1$ and the tree}
	\uIf{$t$ starts with 0}{
		Move $a_l$ and cut $\ell_l$ to modify the two nodes in positions $j_1$ and $j_2$ as follows: $j_1~: (i,j_2, a_r, z, \ell_r)$ and $j_2~: (*, a_l, j_3, t, x')$.\;
	}
	\Else{
		Move $a_l$ and cut $\ell_l$ to modify the two nodes in positions $j_1$ and $j_2$ as follows: $j_1~: (i,j_2, a_r, z, \ell_r)$ and $j_2~: (*, j_3,a_l, t, x')$.\;
	 }
	\tcc{We make this choice so that the left edge is always labeled starting with a 0 and the right edge with a 1}
	% (otherwise we could also put the choice of left and right in superposition, but it would make the structure even more complex).
	\tcc{Since we have moved $j_3$ and $a_l$, the remaining information is: $x$, the modified tree, $j_2$ and the path to $j_1$ (actually the path to $x$ in the new tree)}
	Recompute the invariants on the path to $x$, in reverse order (starting from the address $j_2$). \;
	\tcc{The recomputation of the invariants is reversible (but we still know the path to $x$)}
	Do a lookup of $x$ to uncompute the path to $x$. \;
	\tcc{Now the only information that remains is $x, T(S \cup \{x\})$.}
	\caption{$\INSERT{}$ as a classical algorithm.}\label{algo:insert}
\end{algorithm}

\paragraph{Deletion.}
The deletion can be implemented by uncomputing $\INSERT{}$, since it is a reversible operation. It performs:
\begin{equation}
\INSERT{}^\dagger ~:   \ket{x}\ket{T(S\cup \{x\})} \mapsto \ket{x}\ket{T(S)} \enspace.
\end{equation}
The deletion of an element that is not in $S$ is unspecified. 
%\XB{Paragraphe utile ?}We can check that the operations in~\autoref{algo:insert} can be reverted. First, we compute the path to $x$. Then we update the invariants along this path to take into account the disappearance of $x$. Next, we modify locally the layout of the tree. Doing so, we first obtain an address $j_3$, which contains a leaf. After erasing the leaf (which is reversible), by property of the quantum radix tree, the register that contains $j_3$ is a uniform superposition over all unoccupied cells of the memory. Therefore, we can recompute $\ALLOC{}$ to uncompute the value of $j_3$. Similarly, after having deleted the node at address $j_2$, we erase its address by recomputing $\ALLOC{}$.

\paragraph{Quantum Lookup.}
We can implement a ``quantum lookup'' unitary $\QLOOKUP{}$ which produces a uniform superposition of elements in $S$ having a specific property $P$. The only requirement is that the invariant of nodes has to store the number of nodes in the subtree having this property (and so, leaf nodes will indicate if the given $x$ satisfies $P(x)$ or not).
\begin{equation}
\QLOOKUP{}~: \ket{T(S)}\ket{0} \mapsto \ket{T(S)} \sum_{x \in S | P(x)} \ket{x} \enspace.
\end{equation}

This unitary is implemented by descending in the tree coherently (i.e., in superposition over the left and right paths) with a weight that depends on the number of solutions in the left and right subtrees. First, we initialize an address register $\ket{a}$ to the root. Then, for $n$ times (the maximal depth of the tree), we update the current address register as follows:
\begin{itemize}
\item We count the number of solutions in the left and right subtrees of the node at address $\ket{a}$ (say, $t_l$ and $t_r$).
\item We map $\ket{a}$ to $\ket{a} \left(\sqrt{ \frac{t_l}{t_l + t_r}} \ket{\text{left child of } a} +  \sqrt{ \frac{t_r}{t_l + t_r}} \ket{\text{right child of } a} \right)$. (We do nothing if $\ket{a}$ is a leaf).
%That is, we go to the left with probability $\frac{t_l}{t_l + t_r}$, or to the right with probability $\frac{t_r}{t_l + t_r}$, but coherently.
\end{itemize}
In the end, we obtain a uniform superposition of the paths to all elements satisfying $P$. We can query these elements, then uncompute the paths using an inverse LOOKUP. Likewise, we can also perform a quantum lookup of pairs satisfying a given property, e.g., retrieve a uniform superposition of all collision pairs in $S$.

\subsection{Quantum memory allocators}\label{subsection:memory-alloc}

We now define the unitary $\ALLOC{}$, which given the current state of the memory, creates the uniform superposition of unallocated cells:
\begin{equation}
\ALLOC~: \ket{\text{current memory}} \ket{0} \mapsto \ket{\text{current memory}} \sum_{i \text{ unoccupied}} \ket{i} \enspace.
\end{equation}
We do not need to define a different unitary for un-allocation; we only have to recompute $\ALLOC{}$ to erase the addresses of cells that we are currently cleaning. To implement $\ALLOC{}$, we add to each memory cell a flag indicating if it is allocated. We propose two approaches.

\paragraph{Quantum search allocation.}
Classically, we can allocate new cells by simply choosing addresses at random and checking if they are already allocated or not. Quantumly, we can follow this approach using a \emph{quantum search} over all the cells for unallocated ones. Obviously, for this approach to be efficient, we need the proportion of unallocated cells to be always constant. Besides, if we keep a counter of the number of allocated cells (which does not vary during our quantum walk steps anyway), we can make this operation exact using Amplitude Amplification (Theorem 4 in~\cite{brassard2002quantum}). Indeed, this counter gives the proportion of allocated cells, so we know exactly the probability of success of the amplified algorithm.

We can implement this procedure with a single iteration of quantum search as long as we have a $33\%$ overhead on the maximal number of allocated cells (similarly to the case of searching with a single query studied in~\cite{chi1999quantum}). 
%Moreover, if we are ready to increase the cost of $\ALLOC{}$, we can reduce this overhead to arbitrary constants. For example, with 3 iterates, an overhead of $1/63$ is sufficient. Note that the operations in each iteration involve only a few gates, with a single qRAM gate to access the allocation flag.

\paragraph{Quantum tree allocation.}
 A more standard, but less time-efficient approach to implement $\ALLOC{}$ is to organize the memory cells in a complete binary tree (a heap), so that each node of the tree stores the number of unallocated cells in its children. This tree is not a quantum radix tree, since its size never changes, and no elements are inserted or removed. In order to obtain the uniform superposition of free cell addresses, we mimic the approach of $\QLOOKUP{}$.

%ave efficient allocation is to have an approach similar to the quantum lookup approach. The cells are organized in a binary tree, and each node of the tree stores the number of free cells in its children. Note that contrary to the data structure itself, we do not need to randomize this tree, as its size never changes. With this approach, the overhead is independent of the size of each cell, and is bounded by an overall $4N$ qubits.

\subsection{Higher-level Operations for Collision Walks}

We now implement efficiently the higher-level operations required by our algorithms: performing a quantum walk update ($\SWUP'$), looking for collisions ($\CHECK{}$) and extracting them ($\EXTRACT{}$).

\paragraph{Representation.}
We consider the case of (multi-)collision search. Here the set $S$ is a subset of $[N] = \zo^n$, but we also need to store the images of these elements by the function $f~: \zo^n \to \zo^{m}$. Let $F = \{ f(x) || x, x \in S \}$. A collision of $f$ is a pair $(f(x) || x), (f(y) || y)$ such that $f(x) = f(y)$, i.e., the bit-strings have the same value on the first $m$ bits. 
%More generally, we can define an equivalence relation $\mathcal{R}$ such that $\mathcal{R}(x,y)$ if and only if the first $m$ bits of $x$ and $y$ are equal. 

Since our goal is to retrieve efficiently the collision pairs, we will store both a radix tree $T(S)$ to keep track of the elements, and $T(F)$ to keep track of the collisions. One should note that the sets $F$ and $S$ have the same size. When inserting or deleting elements, we insert and delete both in $T(S)$ and $T(F)$. These trees are stored in two separate chunks of memory cells.

%In the case of collision search, we are looking for collisions in a function $f~: \zo^n \to \zo^{m-n}$. Each identifier is an $m$-bit string formed by concatenating an element of $\zo^n$ and its image: $f(x) | x$. A collision of $f$ is a pair $(f(x) | x), (f(y) | y)$ such that $f(x) = f(y)$, i.e., the bit-strings have the same value on the first $m-n$ bits. More generally, we can define an equivalence relation $\mathcal{R}$ such that $\mathcal{R}(x,y)$ if and only if the first $m-n$ bits of $x$ and $y$ are equal. 

\paragraph{SWUP'.}
We show an efficient implementation of the unitary $\SWUP'$:
\begin{equation}
\SWUP' \ket{T(S)} \ket{T(F)} \ket{c_{S \to S'}} = \ket{T(S')} \ket{T(F')} \ket{c_{S' \to S}}
\end{equation}
Where $c_{S \to S'}$ is the \emph{coin register} which contains information on the transition of a set $S$ to a set $S'$. As we have detailed before, the coin is encoded as a pair $(j,z)$ where $j \in [R]$ is the index of an element in $S$, which has to be removed, and $z \in [N-R]$ is the index of an element in $\zo^n \backslash S$, which has to be inserted. 
%The representation of the 
%coin is independent from the current set $S$; this comes from the regularity of the Johnson graph. 
We implement $\SWUP'$ as follows:
\begin{enumerate}
\item First, we convert the coin register to a pair $x,y$ where: $\bullet$~$y$ is the $z$-th element of $\zo^n$ which is not in $S$ (see details in~\autoref{sec:missing_algo_radix_tree}) and $\bullet$~$x$ is the $j$-th element of $S$ (according to the lexicographic ordering of bit-strings). This can be done easily if the invariant of each node stores the number of leaves in its subtree. Note that both the mapping from $z$ to $y$, and from $j$ to $x$, are reversible. At this point the state is $\ket{T(S)} \ket{T(F)} \ket{x,y}$.
\item We use $\INSERT{}^\dagger$ to delete $x$ from $T(S)$, and delete $f(x)||x$ from $T(F)$.
\item We use $\INSERT{}$ to insert $y$ in $T(S)$ and $f(y)||y$ in $T(F)$. At this point the state is: $\ket{T(S')} \ket{T(F')} \ket{x, y}$ where $S' = (S \backslash \{x\}) \cup \{y\}$ and $F'$ is the set of corresponding images.
\item Finally, we convert the pair $x,y$ back to a coin register.
%We find the index of $x$ in $\zo^n \backslash S'$, and the index of $y$ in $S'$; and we swap the two registers of the coin.
\end{enumerate}

\begin{remark}[Walking in a reduced set]
In our walk, we actually reduce the set of possible elements, due to the previously measured collisions. So the coin does not encode an element of $\zo^n \backslash S$, but of $\zo^n \backslash S \backslash \mathsf{Preim}(C)$, where $C$ is our current table of multicollisions. An adapted algorithm is also given in~\autoref{sec:missing_algo_radix_tree} for this case.
\end{remark}

\paragraph{Checking.}
Checking whether the tree contains a multicollision, or a preimage of some given set, can be made trivial by defining an appropriate invariant of the tree $T(F)$, which counts such solutions. The unitary $\CHECK{}$ simply checks whether this invariant is zero.

\paragraph{Extracting.}
The most important property for our chained quantum walk is the capacity to \emph{extract} multicollisions from the radix tree, in a way that preserves the rest of the state, and allows to reuse a superposition of \emph{marked} vertices for the current walk, as a superposition of \emph{unmarked} vertices for the next one. Recall from~\autoref{sec:new-algo} that we have defined a table of multicollisions $C$, a set $V_R^C$ of sets of size $R$ in $\zo^n \backslash \mathsf{Preim}(C)$, and a set $M_R^C \subseteq V_R^C$ of \emph{marked vertices}, which contain either a new element mapping to $\mathsf{Im}(C)$, or a new collision.

The operation $\EXTRACT{}$ does:
\begin{equation*}
\EXTRACT{} : C, R, \frac{1}{\sqrt{|M_R^C|}} \sum_{ S \in M_R^C} \ket{S} \mapsto C', R', \frac{1}{\sqrt{|  V_{R'}^{C'} \backslash M_{R'}^{C'}|}} \sum_{S \in V_{R'}^{C'} \backslash M_{R'}^{C'}} \ket{S} \enspace,
\end{equation*}
for a smaller $R'$ and a bigger $C'$. It is implemented as~\autoref{algo:full-extraction}.

\begin{algorithm}[htbp]
	\DontPrintSemicolon
	\KwIn{$C$, $R$, uniform superposition over $M_R^C$}
	\KwOut{$C'$ $R'$, uniform superposition over $V_{R'}^{C'} \backslash M_{R'}^{C'}$}
	\textsf{flag} $\leftarrow$ \textsf{true}\;
	$C' \leftarrow C$, $R' \leftarrow R$\;
	\While{ \textsf{flag}  is true }{
		Apply $\CHECK{}$ and measure the result: let \textsf{flag} be the output\;
		\tcc{If \textsf{flag} is true, the superposition has collapsed to a uniform superposition of subsets that contain at least one collision \emph{or} preimage of $C$}
		Perform a ``quantum lookup'' of the solutions (multicollision or preimage)\;
		Select one uniformly at random and copy it outside the tree, with its width $r$\;
		%\tcc{The width is limited $r < \mathsf{poly}(n)$}
		Apply $\INSERT^\dagger$ to remove the elements from the tree\;
		Measure $r$ and these elements\;
		$R' \leftarrow R - r$ \;
		\uIf{ $r > 1$}
		{
		Insert a new entry in $C'$\;		
		}
		\Else{
		Insert the element in $C'$, at the index of its image \;		
		}
%		Move the multicollision 
%		
%		Apply $\EXTRACT{}$ and measure the result\;
%		\uIf{We obtain a multicollision tuple of width $> 1$}
%		{
%			Let $(x_1, \ldots, x_i), u' := f(x_1)$ be the tuple\;
%			$C[u'] = \{ x_1, \ldots, x_i \}$, $U \leftarrow U \cup \{ x_1, \ldots, x_i \}$ \;
%			\tcc{Update the set of images}
%			$U' \leftarrow U' \cup \{u'\}$ \;
%		}
%		\Else{
%			Let $x$ be the element obtained, then $f(x) \in U'$\;
%			Add the obtained element $x$ in the table $C$ at the index of its preimage\;
%		}
	}
	%\Return{$C$, the current state}
	\caption{Multicollision extraction: $\EXTRACT{}$.}\label{algo:full-extraction}
\end{algorithm}

%\YS{Dans l'algo 4, on devrait prendre comme entrée une superposition uniforme de $\mathcal{M}_R^{U, f(U)}$ et sortir une superposition uniforme des noeuds qui ne sont pas dans $\mathcal{M}_{R'}^{U', f(U')}$ avec $U\subset U'$? J'ai l'impression qu'il faut aussi mettre l'étape Check et mesure et update flag à la fin de la boucle while.}

The correctness of $\EXTRACT{}$ comes from the fact that, when we extract and measure an $r$-collision ($x_1, \ldots, x_r)$ with image $u$, we collapse on the uniform superposition over all sets of size $R-r$ which:
\begin{itemize}
\item do not contain any of $x_1, \ldots, x_r$;
\item do not contain $u$ (otherwise this would have gone into the multicollision).
\end{itemize}
We continue until there is no multicollision to measure anymore, where we are guaranteed that the current state is good to run the next walk.

%The correctness of~\autoref{algo:full-extraction} comes from the definition of $\EXTRACT{}$. When we extract a multicollision with image $u'$, we collapse on the uniform superposition over all sets of size $R-i$ taken among the elements of $\zo^n \backslash U$ which \emph{do not map to $u'$}. Thus:
%\begin{itemize}
%\item We add the multicollision tuple $(x_1, \ldots, x_i)$ to $U$ (the forbidden inputs)
%\item We add the measured image $u'$ to the set of images $U'$ (the forbidden images)
%\end{itemize}
%Later on, any vertex that contains an element mapping to $U'$ is considered as marked. This ensures that we indeed collapse on a superposition of \emph{unmarked} vertices. When we find an element that maps to $U'$, we extend one of our current multicollision tuples, which is as good as finding a new collision pair. When we finally measure 0, we are guaranteed that the current state has collapsed to a uniform superposition of unmarked vertices for the next walk. The size of these vertices depends on the number of elements that we extracted.

%
%
%\paragraph{Extracting to a Smaller Node.}
%An important property of $\EXTRACT{}$ is that, when performed on a uniform superposition of marked nodes, it will allow us to create a uniform superposition of \emph{unmarked} nodes for another walk. A simple way to do this is to measure all collisions in the current superposition, which makes it collapse to a superposition of vertices without collisions, and gives us a list of one or more multicollision tuples. This is

\paragraph{Extraction without Measurement.}
\autoref{algo:full-extraction} contains measurements, but it is possible to perform the whole chained quantum walk without. The idea is to apply a sequence of a fixed number of walks, controlled by the current result of $\CHECK{}$. That is, if the current vertex does not contain a solution anymore, we start walking again, but if the vertex still contains a solution, we remove it instead. We are ensured that each of these operations produces a collision, though we do not know which ones did. We also keep track of the current vertex size to implement correctly the walk. At each step, it is reduced at least by 2, and at most by $\mathsf{poly}(n)$ (the maximal collision size). Since the walk step is done with phase estimation, we simply set the precision of the phase estimation circuit at the highest level required, i.e. depending on the initial vertex size, and it will work correctly for all walks.

%We implement a circuit performing (controllably) the maximal number of walk steps that we would expect at each step, which corresponds to the smallest vertex size.

%is probabilistic. In some cases, we might want to be more subtle, and extract without performing any measurement. It is possible by somehow merging 

%The key is to follow a similar approach where we apply the next walk immediately, controlled on the current $\CHECK{}$ result. That is, if the vertex does not contain a (multi)collision (or preimage of $U$) anymore, we start walking again, but if the vertex still contains a collision, we remove it instead. Thus we apply a sequence of controlled circuits applying either $\EXTRACT{}$, or a walk followed by $\EXTRACT{}$. We are ensured that each of these circuits produces a collision (though we do not measure it directly). We also keep track of the current vertex size to implement correctly the walk. At each step, it is reduced at least by 2, and at most by $\mathsf{poly}(n)$ (the maximal collision size). We implement a circuit performing (controllably) the maximal number of walk steps that we would expect at each step, which corresponds to the smallest vertex size.

% !TeX root = ../collision-walk.tex

\section{Searching for Many Collisions, in General}\label{subsection:general-case}

As we have seen, our new algorithm is valid (and tight) for all values of $n$, $m$ and $k \leq 2n-m$ such that $k \leq \frac{m}{4}$. Two approaches can be used for higher values of $k$.

\paragraph{BHT.}
A standard approach to find multiple collisions, which works when $m$ is small, is to extend the BHT algorithm~\cite{DBLP:conf/latin/BrassardHT98}. We select a parameter $\ell$, then make $2^\ell$ queries, and look for $2^k$ collisions on this list of queries. This is done by a quantum search on $\zo^n$ for an input colliding with the list.

There are on average $2^{2n-m}$ collision pairs in the function, so a random element of $\zo^n$ has a probability $\bigO{ 2^{n-m} }$ to be in a collision pair. This gives $\bigO{2^{\ell - m + n}}$ collision pairs for the initial list.

Thus, a search for a collision with the list has $\bigO{2^{\ell - m + n}}$ solutions in a search space of size $2^n$, and it requires $\sqrt{\frac{2^n}{2^{\ell + m-n}}} = 2^{(m-\ell)/2}$ iterates.

If this procedure is to output $2^k$ collisions, we need $\ell$ such that $2^{\ell - m + n} \geq 2^k$ i.e. $\ell -m + n \geq k$. By trying to equalize the complexity of the two steps, we obtain: $\ell = k + \frac{m-\ell}{2} \implies \ell = \frac{2k}{3} + \frac{m}{3}$ which is only valid for $k \leq 3n - 2m$. For a bigger $k$, we can repeat this. We find $2^{3n-2m}$ collisions in time (and memory) $2^{2n-m}$, and we do this $2^{k - (3n-2m)}$ times, for a total time $\bigOt{2^{k + m -n}}$. If we want to restrict the memory then we obtain the tradeoff of $\bigOt{2^{k+m/2-\ell/2}}$ time using $\bigO{2^\ell}$ memory.

\paragraph{Using our method.}
If $k > m/4$, then the memory limitation in \autoref{thm:walk-tradeoff} on $\ell$ becomes relevant. In that case, as we are restricted to $\ell \leq m/2$, the minimal achievable time is $\bigOt{2^{k+m/2-\ell/2}} = \bigOt{2^{k+m/4}}$.

\paragraph{Conclusion.}
The time and memory complexities of the problem are the following (in $\log_2$ and without polynomial factors):
\begin{itemize}
\item If $k \leq 3n-2m$: $\frac{2k}{3} + \frac{m}{3}$ time and memory using BHT
\item Otherwise, if $k \leq \frac{m}{4}$: $\frac{2k}{3} + \frac{m}{3}$ time and memory using our algorithm
\item Otherwise, if $m \leq \frac{4}{3} n$: $k + m -n$ time and $2n-m$ memory using BHT
\item Otherwise, if $m \geq \frac{4}{3} n$: $k + \frac{m}{4}$ time and $\frac{m}{2}$ memory using our algorithm
\end{itemize}
This situation is summarized in~\autoref{fig:parameter-ranges}, and it allows us to conclude:

\begin{theorem}\label{thm:all-ranges}
Let $f : \zo^n \to \zo^m$, $n \leq m \leq 2n$ be a random function. Let $k \leq 2n-m$. There exists an algorithm finding $2^k$ collisions in $\bigOt{ 2^{C(k, m,n)}}$ Clifford+T+\textsf{qRAMR}+\textsf{qRAMW} gates, and using $\bigOt{ 2^{C(k, m,n)}}$ quantum queries to $f$, where:
\begin{equation}
C(k,m,n) = \max\left( \frac{2k}{3} + \frac{m}{3} , k + \min \left( m-n , \frac{m}{4} \right) \right) \enspace.
\end{equation}
\end{theorem}

\begin{proof}
We check that: $k \leq 3n-2m \iff \frac{2k}{3} + \frac{m}{3} \geq k + m-n$ and $k \leq \frac{m}{4} \iff \frac{2k}{3} + \frac{m}{3} \geq k + \frac{m}{4}$. \qed
\end{proof}

We conjecture that the best achievable complexity is, in fact, $C(k,m,n) = \frac{2k}{3} + \frac{m}{3}$ for any admissible values of $k$, $m$ and $n$. It would however require a non-trivial extension of our algorithm, capable of outputting collisions at a higher rate than what we currently achieve.

In terms of time-memory trade-offs, we can summarize the results as:

\begin{theorem}[General Time-memory tradeoff]\label{thm:tradeoff}
For all $k \leq \ell \leq \min(2k/3+m/3,\max(2n-m, m/2))$, there exists an algorithm that computes $2^k$ collisions using $\bigOt{2^\ell}$ qubits and $\bigOt{2^{k+m/2-\ell/2}}$ Clifford+T+\textsf{qRAMR}+\textsf{qRAMW} gates and quantum queries to $f$.
\end{theorem}

Similarly, as in \cite{DBLP:conf/tqc/HamoudiM21}, we conjecture that the trade-off should be achievable for all $\ell \leq 2k/3+m/3$.

%label names, in a key-value dictionary
\newdict{\xlabel}{1}
\setkey{\xlabel}{1}{4/3}
\setkey{\xlabel}{2}{3/2}
\setkey{\xlabel}{3}{2}

\newdict{\ylabel}{0}
\setkey{\ylabel}{1}{1/4}
\setkey{\ylabel}{2}{1/3}
\setkey{\ylabel}{3}{2/5}
\setkey{\ylabel}{4}{2/3}
\setkey{\ylabel}{5}{1}

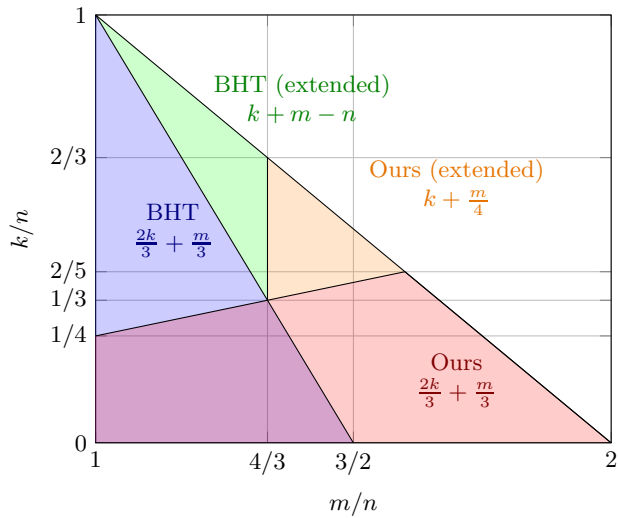
\begin{figure}[htbp]
\centering
\begin{tikzpicture}
\begin{axis}[
scale=1., legend pos=outer north east,
xlabel={$m/n$}, ylabel={$k/n$},
xtick={1,1.333333,1.5,2},
ytick={0,0.25,0.33333,0.4,0.66666,1},
xticklabel={$\xlabel{\ticknum}$},
yticklabel={$\ylabel{\ticknum}$},
ymin=0,ymax=1.0,xmin=1,xmax=2.0, xmajorgrids,ymajorgrids,
title={},
legend style={cells={anchor=west},name=legend,at={(1.1,0.5)},anchor=west}
]
\addplot[name path=limit] coordinates { (1,1) (2,0)};
\addplot[name path=bht] coordinates { (1,1) (1.5,0) };
\addplot[name path=nous] coordinates { (1,0.25) (8/5,8/20) (2,0) };
\addplot[name=vertical] coordinates { (1.333333, 0.333333) (1.333333, 0.666666) };
\addplot[name path=tmp] coordinates {(0,0) (2,0)};
\path[name path=axis] (0,0) -- (2,0);
\addplot [ thick, color=blue, fill=blue,  fill opacity=0.2]
    fill between[ of= bht and tmp];
\addplot [ thick, color=red, fill=red,  fill opacity=0.2 ]
    fill between[ of= nous and tmp];
\addplot [ thick, color=green, fill=green,  fill opacity=0.2 ]
    fill between[ of= limit and bht, soft clip={domain=1:4/3}];
\addplot [ thick, color=orange, fill=orange,  fill opacity=0.2 ]
    fill between[ of= limit and nous, soft clip={domain=4/3:2}];
\node[color=darkblue] at (axis cs:1.15,0.5) {\begin{tabular}{c}BHT \\ $\frac{2k}{3} + \frac{m}{3}$ \end{tabular}};
\node[color=darkred] at (axis cs:1.7,0.15) {\begin{tabular}{c}Ours \\ $\frac{2k}{3} + \frac{m}{3}$ \end{tabular}};
\node[color=darkorange] at (axis cs:1.7,0.6) {\begin{tabular}{c}Ours (extended) \\ $k + \frac{m}{4}$ \end{tabular}};
\node[color=darkgreen] at (axis cs:1.4,0.8) {\begin{tabular}{c}BHT (extended) \\ $k + m - n$ \end{tabular}};
\end{axis}
\end{tikzpicture}
\caption{Exponent in the algorithm depending on the relative values of $k, m$ and $n$.}
\label{fig:parameter-ranges}
\end{figure}

% !TeX root = ../collision-walk.tex

\newcommand{\norm}[1]{||#1||_2}
\newcommand{\sv}{\mathbf{s}}
\newcommand{\cv}{\mathbf{c}}
\section{Applications}\label{section:applications}

In this section, we show how our algorithm can be used as a building block for lattice sieving and to solve the limited birthday problem. We also discuss the problem of multicollision search.

\subsection{Improvements in quantum sieving for solving the Shortest Vector Problem}
\label{sec:quantum_sieving}
\subsubsection{Context}
A lattice $\cL = \cL(\vec{b}_1, \ldots, \vec{b}_d) := \{\sum_{i=1}^d z_i \vec{b}_i : z_i \in \mZ\}$ is the set of all integer combinations of linearly independent vectors $\vec{b}_1,\dots,\vec{b}_d \in \mR^d$.
We call $d$ the \emph{rank} of the lattice and $(\vec{b}_1, \ldots, \vec{b}_d)$ a \emph{basis} of the lattice.

The most important computational problem on lattices is the Shortest Vector Problem (SVP). 
Given a basis for a lattice $\cL \subseteq \mR^d$,
SVP asks to compute a non-zero vector in $\cL$ with the smallest Euclidean norm. 

The main lattice reduction algorithm used for lattice-based cryptanalysis is the famous BKZ algorithm \cite{DBLP:journals/tcs/Schnorr87}. It internally uses an algorithm for solving (near) exact SVP in lower-dimensional lattices. Therefore, finding faster algorithms to solve exact SVP is critical to choosing security parameters of cryptographic primitives. 

Previously, the fastest quantum algorithm solved SVP under heuristic assumptions in $2^{0.2570d + o(d)}$ time~\cite{CL21}. It applies the MNRS quantum walk technique to the state-of-the-art classical algorithm called lattice sieving, where we combine close vectors together to obtain shorter vectors at each step.

It was noted in~\cite{CL21} that the algorithm could be slightly improved if we could find many marked vertices in a quantum walk without repaying the setup each time, which is exactly what we showed in~\autoref{sec:many_collisions}. Without going through the whole framework of the~\cite{CL21} algorithm, we present its main parameters and ideas, and how our quantum walks improves it.

The sieving algorithm works as follows: we start from $N \approx 2^{0.2075d + o(d)}$ points $x_1,\dots,x_n$ on the $d$-dimensional sphere of some norm $R$ and we want to find $N$ pairs $(x_i,x_j)$ such that the norm of $x_i - x_j$ is slightly smaller than $R$. This is one sieving step and the full algorithm performs the above sieving step $poly(d)$ times, so we concentrate on the running time of a single sieving step.

\subsubsection{Parameters of the algorithm.} We fix a dimension $d$. The algorithm uses a free parameter $c \in (0,1)$. 
 For an angle $\alpha \in [0,\pi/2]$, $V_d(\alpha)$ is the ratio of the volume of a spherical cap of angle $\alpha$ to the volume of the $d$-dimensional sphere. This means 
 $$ V_d(\alpha) = \poly(d)\sin^d(\alpha).$$
 We also define:
\begin{itemize}
	\item $N = \frac{1}{V_d(\pi/3)} \approx 2^{0.2075d + o(d)}$.
	\item $\alpha$ st. $V_d(\alpha) = N^{-(1-c)}$.
	\item $\theta^*_\alpha = 2\arcsin(\frac{1}{2\sin(\alpha)})$.
	\item $\zeta$ st. $N^\zeta = N^{2c}V_{d-1}(\theta^*_\alpha)$.
\end{itemize}
The quantum algorithm of~\cite{CL21} in dimension $d$ with free parameter $c$ runs in time
\begin{align}\label{Eq:SievingComplexity}
	T = NBREP \cdot(INIT + N^{1-c}FAS_1).
\end{align}
where 
\begin{itemize}
	\item $NBREP = \max\{1,N^{c - \zeta + o(1)}\}$.
	\item $INIT = N^{1 + o(1)}$.
	\item $FAS_1$ is the running time of finding many marked elements in a Johnson graph using a quantum walk, which we will describe more in detail below.
\end{itemize}
The idea of the FindAllSolutions subroutine ($FAS_1$) is the following: we start from $N^c$ points $x_1,\dots,x_{N^c}$ of norm $R$ which already are at an angle $\alpha$ of a certain point $\sv$, and we want to find most (at least a constant fraction) of the pairs $(x_i,x_j)$ st. $\norm{x_i - x_j} < R$. There are on average $N^\zeta$ and the goal of this procedure is to find a constant fraction of them. 

\subsection{Analysis of $FAS_1$}

The analysis of this random walk involves a new free parameter $c_1 < c$ over which we can optimize. Following~\cite{CL21}, we also define 
\begin{itemize}
	\item $\beta$ st. $V_d(\beta) = \frac{1}{N^{c_1}}$.
	\item $\rho_0$ st. $N^{\rho_0} = \frac{V_d(\beta)}{W_d(\beta,\theta^*_\alpha)}, \ $ where $W_d(\beta,\theta^*_\alpha) = \mathrm{poly}(d)\cdot \left(1 - \frac{2 \cos^2(\beta)}{1 + \cos(\theta^*_\alpha)}\right)^{d/2}$.
\end{itemize}

In order to find these solutions, the authors of~\cite{CL21} construct a code $C$ on the sphere and check only the pair $(x_i,x_j)$ if $x_i,x_j$ are at angle at most $\beta$ of the same code point $\cv$. This means they have a function $f$ that maps a point $x_i$ to the nearest codeword, which in this setting is efficiently computable. Then the idea is to look for solution pairs $(x_i,x_j)$ st. $f(x_i) = f_(x_j)$. By doing so, they miss on some collision pairs, but there will be solutions that satisfy this property and will be easy to find. Then once they run out of solutions of this form, they choose another code $C$ and start again. Here, the code will be of size $2^{c_1}$.

To perform the above, they use a quantum walk for collision finding, except that pairs $(x_i,x_j)$ st. $f(x_i) = f(x_j)$ do not necessarily satisfy $\norm{x_i - x_j} < R$ (but this condition can also be checked efficiently). They construct the same Johnson graph as for collision finding. Each node contains $N^{c_1}$ points for a parameter $c_1 < c$ and $2$ nodes are adjacent if they differ by exactly one point. The only difference is that a node is marked not only if it contains $x_i,x_j$ such that $f(x_i) = f(x_j)$, but also that $\norm{x_i - x_j} < R$. For each node, the probability that a node contains a solution pair $(x_i,x_j)$ is $N^{2c_1}V_{d-1}(\theta^*_\alpha)$ and the probability that it also satisfies $f(x_i) = f(x_j)$ is $N^{-\rho}$, so $\eps = N^{2c_1 - \rho_0} V_{d-1}(\theta^*_\alpha)$. On the other hand, looking only at pairs such that $f(x_i) = f(x_j)$ allows to perform the quantum walk with efficient update, as for the regular collision finding. 

This quantum walk has parameters:\footnote{In~\cite{CL21}, there are extra parameters $c_2,\rho \approx 0$, we perform the same choice here (we checked that with our improvement, this remains the optimal choice).} 
\begin{itemize}
	\item $S = N^{c_1}$.
	\item $\delta = N^{-c_1}$.
	\item $\eps = N^{2c_1 - \rho_0} V_{d-1}(\theta^*_\alpha)$.
	\item $U = 1$.
	\item $C = 1$.
\end{itemize}
For each choice of $C$, we need to find $N^{\zeta - \rho_0}$ random marked vertices, and then repeat this $N^{\rho_0}$ times to find $N^{\zeta}$ solutions. The formula used in~\cite{CL21} is 
$$ FAS_1 = N^{\rho_0} \cdot N^{\zeta - \rho_0} \left(S + \frac{1}{\sqrt{\eps}}\left(\frac{1}{\sqrt{\delta}}U + C\right)\right).$$
However, with our results, we don't have to redo the setup in the quantum walk and we obtain 
$$ FAS_1 = N^{\rho_0} \cdot  \left(S + \frac{N^{\zeta - \rho_0} }{\sqrt{\eps}}\left(\frac{1}{\sqrt{\delta}}U + C\right)\right).$$
With this improvement, we redid the optimization of~\cite{CL21} and obtained the following new running for quantum sieving. 

We take the following set of parameters: $c \approx 0.3875, c_1 \approx 0.27$ which gives $\zeta \approx 0.1568$ and $\rho_0 \approx 0.1214$. Notice that with these parameters, we are indeed in the range of Theorem~\ref{thm:main-result} since the number of solutions we extract is $2^k = N^{\zeta - \rho_0} \approx N^{0.0354}$ and the range of the function $f$ on which we collision is $2^m = 2^{c_1} \approx N^{0.27}$ (the number of points in the code), so we indeed have $k \le \frac{m}{4}$. The parameters of the quantum walk become:
$$ S \approx N^{0.27}, \ \eps \approx N^{-0.2}, \ \delta \approx N^{-0.27}, \ U = C = 1 \enspace. $$  
This gives $FAS_1 \approx N^{0.27}$. Plugging this into Equation~\ref{Eq:SievingComplexity}, we get a total running time of $N^{1.2347}$ which is equal to  $2^{0.2563d + o(d)}$ (recall that $N = \frac{1}{V_d(\pi/3)} \approx 2^{0.2075d + o(d)}$) improving slightly the previous running time of $2^{0.2570d + o(d)}$.

\subsection{Solving the Limited Birthday Problem}

The following problem is very common in symmetric cryptanalysis. It appears for example in impossible differential attacks~\cite{DBLP:conf/asiacrypt/BouraNS14}, but also in rebound distinguishers~\cite{DBLP:journals/tosc/HosoyamadaNS20}. In the former case we use generic algorithms to solve the problem for a black-box $E$, and in the latter, a valid distinguisher for $E$ is defined as an algorithm outputting the pairs faster than the generic one.

\newcommand{\Din}{\mathcal{D}_{\mathrm{in}}}
\newcommand{\Dout}{\mathcal{D}_{\mathrm{out}}}
\newcommand{\Deltain}{\Delta_{\mathrm{in}}}
\newcommand{\Deltaout}{\Delta_{\mathrm{out}}}

\begin{problem}[Limited Birthday]\label{pb:lb}
Given access to a black-box permutation $E~: \zo^n \to \zo^n$ and possibly its inverse $E^{-1}$, given two vector spaces $\Din$ and $\Dout$ of sizes $2^{\Deltain}$ and $2^{\Deltaout}$ respectively, find $2^k$ pairs $x,x'$ such that $x \neq x', x \oplus x' \in \Din, E(x) \oplus E(x') \in \Dout$.
\end{problem}

For simplicity, we will focus only on the time complexity of the problem, although some parameter choices require a large memory as well. Classically the best known time complexity is given in~\cite{DBLP:conf/asiacrypt/BouraNS14}:
\begin{equation}\label{eq:classical-lb}
\max \left( \min_{\Delta \in \{ \Deltain, \Deltaout \}} \left( \sqrt{2^{k + n+1 - \Delta}} \right), 2^{k + n+1 - \Deltain - \Deltaout}  \right)
\end{equation}

This complexity is known to be tight for $N = 1$~\cite{DBLP:journals/tosc/HosoyamadaNS20}.

In the quantum setting, we need to consider superposition access to $E$ and possibly $E^{-1}$ to have a speedup on this problem\footnote{When $E$ is a black box with a secret key, this is the \emph{Q2} model, see e.g.\cite{DBLP:journals/tosc/KaplanLLN16}. In some cases, e.g. rebound distinguishers, $E$ does not contain any secret.}. Previously the methods used~\cite{DBLP:journals/tosc/KaplanLLN16} involved only individual calls to Ambainis' algorithm (when there are few solutions) or an adaptation of the BHT algorithm (when there are many solutions).

The quantum algorithm, as the classical one, relies on the definition of \emph{structures} of size $2^{\Deltain}$, which are subsets of the inputs of the form $T_x = \{ x \oplus v, v \in \Din\}$ for a fixed $x$. For a given structure $T_x$, we can define a function $h_x~: \zo^{\Deltain} \to \zo^{n - \Deltaout}$ such that any collision of $h_x$ yields a pair solution to the limited birthday problem. The algorithm then depends on the number of required pairs compared to the (expected) number of collisions of $h_x$.
\begin{itemize}
\item If $2^k < \frac{2^{2\Deltain}}{2^{n-\Deltaout}} \iff k < 2\Deltain - n + \Deltaout$, then we need only one structure. To recover all the pairs, we need a time exponent (by~\Cref{thm:all-ranges}):
\[ \max \left( \frac{2k}{3} + \frac{n - \Deltaout}{3}, k + \min\left(  n-\Deltaout - \Deltain, \frac{n-\Deltaout}{4} \right)  \right) \]
%: all the pairs are recovered  using our multiple collision-finding algorithm, in time and memory $\bigOt{N^{2/3}2^{\frac{n-\Deltaout}{3}}}$.
\item If $\frac{2^{2\Deltain}}{2^{n-\Deltaout}} < 1$, then we follow the approach of~\cite{DBLP:journals/tosc/KaplanLLN16}, which is to repeat $2^k$ times a Grover search among structures, to find one that contains a pair (this is done with Ambainis' algorithm). The time  exponent is $k + \frac{n-{\Deltaout}}{2} - \frac{\Deltain}{3}$.
% and QRAQM usage is $2^{\frac{2}{3}\Deltain}$.
\item If $1 < \frac{2^{2\Deltain}}{2^{n-\Deltaout}} < 2^k $, we need to consider several structures and to extract all of their collision pairs. Using~\Cref{thm:all-ranges} this gives a time exponent:
\[ \max \left( k + \frac{2}{3}(n-\Deltain -\Deltaout) ,k +  \min\left(  n-\Deltaout - \Deltain, \frac{n-\Deltaout}{4} \right) \right)  \]
%For each structure, we use our algorithm to extract $N' = \frac{2^{2\Deltain}}{2^{n-\Deltaout}}$ collisions in time $N'^{\frac{2}{3}}2^{\frac{n - \Deltaout}{3}}$. We repeat this as many times as necessary to obtain a total of $N$ pairs. This gives a time $N2^{\frac{2}{3}(n-\Deltain -\Deltaout)}$.
\end{itemize}

Finally, we can swap the roles of $\Deltain$ and $\Deltaout$ and take the minimum. Unfortunately this does not lead to an equation as simple as~\autoref{eq:classical-lb}.

%\AS{difficult to have a global formula here. Before it was OK but now it's difficult.}

%By putting all these cases together, we obtain the following formula for the quantum limited birthday problem.
%
%\begin{theorem}
%Given quantum oracle access to $E$ and $E^{-1}$, ~\Cref{pb:lb} can be solved in quantum time:
%  \begin{multline}
%\max \bigg( \min_{\Delta \in \{ \Deltain, \Deltaout \} } N^{\frac{2}{3}}2^{\frac{n-\Delta}{3}},\\
%     \min \left( N2^{\frac{2}{3} (n-\Deltaout-\Deltain)} , \min_{\Delta \in \{ \Deltain, \Deltaout\} } N2^{\frac{n}{2} - \frac{\Delta}{6}- \frac{\Deltain + \Deltaout}{3}} \right)  \bigg) \enspace.
%  \end{multline}
%\end{theorem}

% !TeX root = ../collision-walk.tex

\subsection{On multicollision-finding}
A natural extension of this work would be to look for multicollisions.

\begin{problem}[$r$-collision search]
Let $f~: \zo^n \to \zo^m$ be a random function. Find an $r$-collision of $f$, that is, a tuple $(x_1, \ldots, x_r)$ of distinct elements such that $f(x_1) = \ldots = f(x_r)$.
\end{problem}

As with collisions, the lower bound by Liu and Zhandry~\cite{DBLP:conf/eurocrypt/LiuZ19} is known to be tight when $m \leq n$. The corresponding algorithm is an extension of the BHT algorithm which constructs increasingly smaller lists of $i$-collisions, starting with $1$-collisions (evaluations of the function $f$ on arbitrary points) and ending with a list of $r$-collisions. 

This algorithm, given in~\cite{DBLP:conf/pqcrypto/HosoyamadaSTX19,DBLP:journals/tcs/HosoyamadaSTX20}, finds $2^k$ $r$-collisions in time and memory:

\[ \bigOt{2^{k\frac{2^{(r-1)}}{2^r-1}} 2^{m\frac{2^{(r-1)}-1}{2^r-1}}}  \enspace. \]

%As with collisions, we know a lower bound, which is known to be tight when $m\leq n$. The principle of the algorithm is simple: it constructs increasingly smaller lists of i-collisions, and use the list $i$ to search for an $i+1$-collision, like in the BHT algorithm.
As with 2-collisions, it is possible to extend it when $m > n$. Of course, there's a constraint: the list $i$ must contain more tuples that are part of an $i+1$-collision than the size of the list $i+1$.

The size of each $i$-collision list is $N_i = 2^{k\frac{2^r-2^{r-i}}{2^r-1}} 2^{m \frac{2^{r-i}-1}{2^r-1}}$. The probability that an $i$-collision extends to an $i+1$-collision is of order $2^{n-m}$. Hence, for the algorithm to work, we must have, for all $i$,
$N_{i+1}/N_i \leq 2^{n-m}$. This means:
%Zhandry's algorithm finds $2^s$ $k$-collisions in time and memory $2^{s(2^{k-1}/(2^k-1)}2^{m(2^{k-1}-1)/(2^k-1)}$. The size of each list of $i$-collision is $N_i = 2^{s(2^k-2^{k-i})/(2^k-1)}2^{m(2^{k-i}-1)/(2^k-1)}$. The probability that an $i$-collision happens to be part of an $i+1$-collision is $2^{n-m}$. Hence, for the algorithm to work, we must have, for all $i$,
%$N_{i+1}/N_i \leq 2^{n-m}$. This means
\[
k \frac{2^{r-i-1}}{2^r-1} - m\frac{2^{r-i-1}}{2^r-1} \leq n-m \enspace.
\]
This constraint is the most restrictive for the largest possible $i$, $r-1$. We obtain the following constraint, which subsumes the others:
\[
 k \frac{1}{2^r-1}+m\left(1-\frac{1}{2^r-1} \right) \leq n \enspace.
\]

This gives the point up to which this algorithm meets the lower bound.
%Hence, as with 2-collisions, it is possible to extend the range of the BHT-like algorithm, up to a point.
We could use our new algorithm as a subroutine in this one, to find 2-collisions, and this would allow to relax the constraint over $N_2/N_1$. Unfortunately, this cannot help to find multicollisions, as the other constraints are more restrictive. More generally, these constraints show that it is not possible to increase the range of the BHT-like $r$-collision algorithm solely by using an $r-i$-collision algorithm with an increased range.

%\section{Conclusion}
%
%A good conclusion happens here.

\ifeprint
\subsubsection*{Acknowledgments.}
A.S. wants to thank Nicolas David and María Naya-Plasencia for discussions on the limited birthday problem. A.S. is supported by ERC-ADG-ALGSTRONGCRYPTO (project 740972). Y.S. is supported by EPSRC grant EP/S02087X/1 and EP/W02778X/1.

\fi

%\nocite{*}
\bibliography{biblio}
\bibliographystyle{splncs03}
%\printbibliography

\newpage
\appendix
\ifeprint
\section*{Appendix}
\else
\section*{Supplementary Material}
\fi

\section{$z$-th element outside the radix tree}
\label{sec:missing_algo_radix_tree}

\SetKwFunction{FindNthNotInSubtree}{FindNthNotInSubtree}
\SetKwFunction{Left}{left}
\SetKwFunction{Right}{right}
\SetKwFunction{Leaves}{leaves}
\SetKwFunction{Subtree}{subtree}
\SetKwFunction{FindNthNotInTree}{FindNthNotInTree}
\SetKwFunction{Root}{root}
\SetKwFunction{CountInIntervalNotTree}{CountInIntervalNotTree}
\SetKwFunction{CountInIntervalNotSubtree}{CountInIntervalNotSubtree}
\SetKwFunction{FindNthNotInTwoSubtrees}{FindNthNotInTwoSubtrees}
\SetKwFunction{FindNthNotInTwoTrees}{FindNthNotInTwoTrees}

%\AS{this could go into main body}

In this section, we solve the following problem:
\begin{quote}
	Find the value $y$ of the $z$-th element of $\zo^n$ which is not in $S$.
\end{quote}
We need the following invariant in the tree: each node $n$ of $S$ stores
the number of leaves in the subtree rooted at $n$. We denote this quantity
by $\mathrm{leaves}(n)$.
See Figure~\ref{fig:tree-logical_with_leave_count} for an example.

\begin{figure}[htbp]
	\centering
	\begin{tikzpicture}[level distance=10mm, nodes={draw,rectangle,rounded corners=.2cm,->}, level 1/.style={sibling distance=40mm}, level 2/.style={sibling distance=15mm}, level 3/.style={sibling distance=15mm}]
		\node[draw, circle] {5}
		child { node[draw, circle] {2} 
			child{ node {$\mathtt{0000}$} edge from parent node[draw=none, left] {\texttt{00}} }
			child{ node {$\mathtt{0010}$} edge from parent node[draw=none, left] {\texttt{10}} }
			edge from parent node[draw=none, left] {\texttt{00}}
		}
		child { node[draw, circle] {3}
			child { node[draw, circle] {2} 
				child{ node {$\mathtt{1001}$} edge from parent node[draw=none, left] {\texttt{01}} }
				child{ node {$\mathtt{1011}$} edge from parent node[draw=none, left] {\texttt{11}} }
				edge from parent node[draw=none, left] {\texttt{0}}
			}
			child { node {$\mathtt{1111}$} edge from parent node[draw=none, left] {\texttt{111}}}
			edge from parent node[draw=none, left] {\texttt{1}}
		};
	\end{tikzpicture}
	\caption{Example of tree where each node stores the number of leaves in the subtree. We omit this quantity (which is 1) on the leaves themselves for readability.\label{fig:tree-logical_with_leave_count}}
\end{figure}
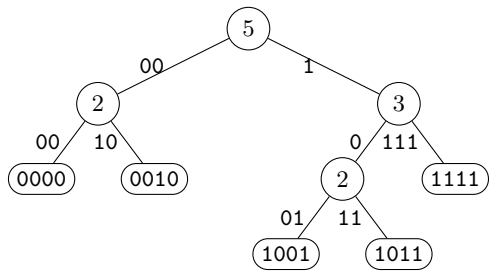

Assume that $S$ has $R$ elements in $\set{0,1}^n$, let $N=2^n$.
Assume here that we have some easily computable order on $\set{0,1}^n$,
represented by a map $\phi:\set{0,1}^n\to N$ that assigns to each bit-string
its index, and its inverse $\phi^{-1}$ also easily computable.
Given $i\in[N-R]$, the goal is to find the $i^{th}$ element in $\set{0,1}^n\setminus S$.

\newpage

\begin{algorithm}[H]
	\DontPrintSemicolon
	\KwIn{index $i$, radix tree $T$, node \texttt{node}}
	\KwOut{$i^{th}$ element in 
		$\set{x\in\set{0,1}^n:\phi(x)\geqslant\phi(\ell)}\setminus \Subtree{T,\texttt{node}}$
		where $\ell$ is the bit-string of the left-most leaf in the
		subtree rooted at \texttt{node}
	}
	\eIf{\texttt{node} is a leaf}{
		Let $x$ be the bit-string of \texttt{node} \;
		\Return{$\phi^{-1}(\phi(x)+i)$}
	}{
		\tcc{Here \texttt{node} must have two children \Left{\texttt{node}}
			and \Right{\texttt{node}}}
		Let $x$ be the bit-string of the left-most leaf in subtree
		rooted at \texttt{node} \;
		Let $y$ be the bit-string of the left-most leaf in subtree
		rooted at \Right{\texttt{node}} \;
		\tcc{compute the number of elements in $[x,y)\setminus T$}
		$\delta\gets \phi(y)-\phi(x)-\Leaves{\Left{\texttt{node}}}$ \;
		\eIf{$i\geqslant \delta$}{
			\Return{\FindNthNotInSubtree{$i-\delta$,T,\Right{\texttt{node}}}}
		}{
			\Return{\FindNthNotInSubtree{$i$,T,\Left{\texttt{node}}}}
		}
	}
	\caption{FindNthNotInSubtree($i$, $T$, \texttt{node})}
\end{algorithm}

\begin{algorithm}[H]
	\DontPrintSemicolon
	\KwIn{index $i$, radix tree $T$}
	\KwOut{$i^{th}$ element in $\set{0,1}^n\setminus T$}
	Compute the bit-string $x$ of the leftmost leaf of the tree $T$ \;
	\eIf{$i<\phi(x)$}{
		\Return{$\phi^{-1}(i)$}
	}{
		\Return{\FindNthNotInSubtree{$i-\phi(x)$, $T$, \Root{$T$}}}
	}
	\caption{FindNthNotInTree($i$, $T$)}
\end{algorithm}

\begin{theorem}
	$\FindNthNotInTree(i,T)$ returns the $i^{th}$ element in $\set{0,1}^n\setminus T$ in $\poly(n)$ time. 
\end{theorem}

We now consider the same problem where we have two trees $T$ and $T'$ and
we want to find the $i^{th}$ element in $\set{0,1}^n$ which is not in $T$
and not in $T'$. We assume that $T$ and $T'$ have \textbf{disjoint} leaves. This problem appears in our chained quantum walk in \Cref{sec:many_collisions}.

\begin{algorithm}[h]
	\DontPrintSemicolon
	\KwIn{bit-strings $u$ and $v$, radix tree $T$, a node \texttt{node} of $T$}
	\KwOut{size of $[u,v)\setminus \set{\text{all the leaves in the subtree root at \texttt{node}}}$}
	Let $x$ be the bit-string of the left-most leaf in subtree
	rooted at \texttt{node} \;
	Let $y$ be the bit-string of the right-most leaf in subtree
	rooted at \texttt{node} \;
	\tcc{when the interval $[u,v)$ entirely covers the subtree}
	\If{$\phi(u)\leqslant \phi(x)$ and $\phi(y)<\phi(v)$}{
		\Return{$\phi(v)-\phi(u)-\Leaves{\texttt{node}}$}
	}
	\tcc{when the interval $[u,v)$ is disjoint from the subtree}
	\If{$\phi(v)\leqslant\phi(x)$ or $\phi(y)<\phi(u)$}{
		\Return{$\phi(v)-\phi(u)$}
	}
	\tcc{if we are here, \texttt{node} cannot be a leaf}
	Let $z$ be the bit-string of the left-most leaf in subtree
	rooted at \Right{\texttt{node}} \;
	\uIf{$\phi(v)\leqslant\phi(z)$}{
		\tcc{when  $[u,v)$ only intersects the left subtree}
		\Return{\CountInIntervalNotTree{$u$, $v$, $T$, \Left{\texttt{node}}}}
	}
	\uElseIf{$\phi(z)\leqslant\phi(u)$}{
		\tcc{when  $[u,v)$ only intersects the right subtree}
		\Return{\CountInIntervalNotTree{$u$, $v$, $T$, \Right{\texttt{node}}}}
	}
	\Else{
		\Return{\CountInIntervalNotTree{$u$, $z$, $T$, \Left{\texttt{node}}}
			+\CountInIntervalNotTree{$z$, $v$, $T$, \Right{\texttt{node}}}}
	}
	
	\caption{CountInIntervalNotSubtree($u$, $v$, $T$, \texttt{node})}
\end{algorithm}

\begin{algorithm}[h]
	\DontPrintSemicolon
	\KwIn{bit-strings $u$ and $v$, radix tree $T$}
	\KwOut{size of $[u,v)\setminus T$}
	\Return{\CountInIntervalNotSubtree{$u, v, T, \Root{T}$}}
	\caption{CountInIntervalNotTree($u$, $v$, $T$)}
\end{algorithm}

\begin{algorithm}[h]
	\DontPrintSemicolon
	\KwIn{index $i$, radix trees $T$ and $T'$ with disjoint leaves, node \texttt{node} of $T$}
	\KwOut{$i^{th}$ element in 
		$\set{x\in\set{0,1}^n:\phi(x)\geqslant\phi(\ell)}\setminus (\Subtree{T,\texttt{node}}\cup T')$
		where $\ell$ is the bit-string of the left-most leaf in the
		subtree of $T$ rooted at \texttt{node}
	}
	\eIf{\texttt{node} is a leaf}{
		Let $x$ be the bit-string of \texttt{node} \;
		$\delta\gets \CountInIntervalNotTree{$\texttt{0}^n$, $x$, $T'$}$ \;
		\Return{\FindNthNotInSubtree{$i+\delta+1$, $T'$}}
	}{
		\tcc{Here \texttt{node} must have two children \Left{\texttt{node}}
			and \Right{\texttt{node}}}
		Let $x$ be the bit-string of the left-most leaf in subtree
		rooted at \texttt{node} \;
		Let $y$ be the bit-string of the left-most leaf in subtree
		rooted at \Right{\texttt{node}} \;
		\tcc{compute the number of elements in $[x,y)\setminus (T\cup T')$}
		$\delta\gets \CountInIntervalNotTree{x,y,T'}-\Leaves{\Left{\texttt{node}}}$ \;
		\eIf{$i\geqslant \delta$}{
			\Return{\FindNthNotInTwoSubtrees{$i-\delta$,T,T',\Right{\texttt{node}}}}
		}{
			\Return{\FindNthNotInTwoSubtrees{$i$,T,T',\Left{\texttt{node}}}}
		}
	}
	\caption{FindNthNotInTwoSubtrees($i$, $T$, $T'$, \texttt{node})}
\end{algorithm}

\begin{algorithm}[h]
	\DontPrintSemicolon
	\KwIn{index $i$, radix trees $T$ and $T'$ with disjoint leaves}
	\KwOut{$i^{th}$ element in $\set{0,1}^n\setminus (T\cup T')$}
	Compute the bit-string $x$ of the leftmost leaf of the tree $T$ \;
	\tcc{compute the number of elements on the left of $T$ that are not in $T'$}
	$\delta\gets \CountInIntervalNotTree{$\texttt{0}^n,x,T'$}$ \;
	\eIf{$i<\delta$}{
		\Return{\FindNthNotInTree{$i, T'$}}
	}{
		\Return{\FindNthNotInTwoSubtrees{$i-\delta$, $T$, $T'$, \Root{$T$}}}
	}
	\caption{FindNthNotInTwoTrees($i$, $T$, $T'$)}
\end{algorithm}

\begin{theorem}
	$\FindNthNotInTwoTrees($i$, $T$, $T'$)$ returns the $i^{th}$ element in $\set{0,1}^n\setminus (T\cup T')$ in $\poly(n)$ time. 
\end{theorem}

\end{document}